\documentclass[a4paper,11pt]{article}
\usepackage{amsthm}
\usepackage{cite}
\usepackage{indentfirst}
\usepackage{stmaryrd}

\usepackage[top=1 in,bottom=1 in,left=1.0 in,right=1.0 in]{geometry}

\usepackage{xcolor}
\usepackage[colorlinks,linkcolor=black,anchorcolor=blue,citecolor=blue]{hyperref}

\newtheorem{lemma}{Lemma}[section]

\newtheorem{theorem}{Theorem}[section]
\newtheorem{proposition}{Proposition}[section]
\newtheorem{corollary}{Corollary}[section]
\newtheorem{remark}{Remark}[section]

\newcommand{\bsb}{\begin{subequations}}
\newcommand{\esb}{\end{subequations}}

\newcommand{\ti}[1]{\tilde{#1}}

\usepackage{amsmath}

\usepackage{amsfonts}
\usepackage{latexsym}
\usepackage{amssymb}
\usepackage{appendix}
\usepackage{mathrsfs}
\usepackage{graphicx}
\usepackage{bm}
\usepackage{cases}
\usepackage{enumitem}
\allowdisplaybreaks

\usepackage{color}
\newcommand{\bred}{\begin{color}{red}}
\newcommand{\ecl}{\end{color}}
\newcommand{\bblue}{\begin{color}{blue}}
\newcommand{\bgre}{\begin{color}{green}}
\newcommand{\bora}{\begin{color}{orange}}

\numberwithin{equation}{section}

\title{Constraints of the D$\Delta$KP hierarchy to the semi-discrete
AKNS and Burgers hierarchies}
\author{Jin Liu$^{1}$,~~ Da-jun Zhang$^{1,2}$\footnote{
Corresponding author: djzhang@staff.shu.edu.cn}\\
\small{$^{1}$Department of Mathematics, Shanghai University, Shanghai 200444, China}\\
\small{$^{2}$Newtouch Center for Mathematics of Shanghai University, Shanghai 200444, China}}

\date{\today}
			
\begin{document}
				
\maketitle

\begin{abstract}
The paper investigates three eigenfunction constraints of two (2+1)-dimensional differential-difference
integrable systems.
First, we revisit the known squared eigenfunction symmetry constraint of
the differential-difference Kadomtsev-Petviashvili (D$\Delta$KP) hierarchy,
which gives rise to a semi-discrete Ablowitz-Kaup-Newell-Segur hierarchy.
Second, we introduce a linear eigenfunction constraint for the D$\Delta$KP system
and obtain a combined semi-discrete Burgers (sdBurgers) hierarchy.
In the third one, we consider another linear eigenfunction constraint for
the modified D$\Delta$KP system and obtain the same combined sdBurgers hierarchy.
All these constraint results are proved by using recursive algebraic structures
of the involved integrable hierarchies generated by their master symmetries.

\begin{description}
\item[Keywords:] D$\Delta$KP hierarchy,
semi-discrete AKNS hierarchy,
semi-discrete Burgers hierarchy, eigenfunction constraint, master symmetry

\end{description}
\end{abstract}
				
%
				
\section{Introduction}\label{sec-1}

It is well known that the Kadomtsev-Petviashvili (KP) equation,
\begin{equation}\label{KP}
4u_{xt}=(u_{xxx}+12uu_x)_x+3u_{yy},
\end{equation}
is a representative in $(2+1)$-dimensional integrable systems.
Its hierarchy are characterized by the Lax equation
\begin{equation}\label{Lax-eq}
L_{t_m}=[A_m,L]=A_mL-LA_m
\end{equation}
as the compatibility of
\begin{subequations}\label{Lax-pair}
\begin{align}
& L \varphi=\lambda \varphi, \label{Lax-pair-L}\\
& \varphi_{t_m}=A_m \varphi, \label{Lax-pair-t}
\end{align}
\end{subequations}
where $L$ is the pseudo-differential operator
\begin{equation}
L=\partial_x +u_2 \partial^{-1}_x  +u_3 \partial^{-2}_x+ \cdots,
\end{equation}
$u_j=u_j(x, t_1=x, t_2=y, t_3=t, t_4, \cdots)$,
$A_m=(L^m)_{\geq 0}$ contains the differential part of $L^m$,
e.g., $A_2=\partial_x^2+2u_2$.
The KP hierarchy and the pseudo-differential operator formulation
not only are the building blocks of the universal Sato theory \cite{MJD-book-2000}
but also exhibit very rich mathematical structure in integrable systems, e.g. \cite{K-book-2000}.

As $(2+1)$-dimensional integrable systems, the KP hierarchy shows various links with
 $(1+1)$-dimension integrable equations.
Apart form the well known dimension reductions that give rise to the Korteweg-de Vries (KdV)
and Boussinesq hierarchies,
the KP hierarchy can yield the Ablowitz-Kaup-Newell-Segur (AKNS) hierarchy.
The KP hierarchy admit a ``squared'' eigenfunction symmetry \cite{MSS-JMP-1990}
$\sigma=(\varphi \varphi^*)_x$ where  $\varphi^*$ stands for a solutions of the adjoint form of \eqref{Lax-pair-t}:
\begin{equation}\label{Lax-pair-t*}
\varphi^*_{t_m}=-A^*_m \varphi^*.
\end{equation}
It turns out that the symmetry constraint $u_2=\varphi \varphi^*$
leads \eqref{Lax-pair-L} to the AKNS spectral problem
and converts \eqref{Lax-pair-t} and \eqref{Lax-pair-t*} to the AKNS hierarchy.
This remarkable fact was first revealed in 1991 independently in \cite{CL-PLA-1991} and \cite{KSS-PLA-1991}.
The rigorous proofs were given one year later  in \cite{X-IP-1992} and \cite{CL-JPA-1992},
respectively, from different ways.
The AKNS hierarchy has a recursion operator but the KP hierarchy does not.
Ref.\cite{X-IP-1992} showed that under the constraint $u_2=\varphi \varphi^*$, the
flows $(\varphi_{t_{m+1}}, \varphi^*_{t_{m+1}})^T$ and $(\varphi_{t_{m}}, \varphi^*_{t_{m}})^T$
are connected by the AKNS recursion operator.
On the other hand,  the recursive structures
of the AKNS flows and the KP flows can be described by using their master symmetries.
Ref.\cite{CL-JPA-1992} completed the proof by comparing
the  recursive structures of the AKNS flows and $\{(\varphi_{t_{m}}, \varphi^*_{t_{m}})^T\}$
that are generated by master symmetries.
The  ``squared'' eigenfunction symmetry constraint should be general,
because such constraints and links with  $(1+1)$-dimensional integrable hierarchies
have been found for the modified KP (mKP),
differential-difference KP (D$\Delta$KP),
differential-difference mKP (D$\Delta$mKP) hierarchies.
The obtained $(1+1)$-dimensional integrable systems
are the Chen-Lee-Liu (CLL) hierarchy which is linked to mKP \cite{C-JMP-2002,CL-JPA-1992},
the semi-discrete AKNS (sdAKNS) hierarchy which is linked to D$\Delta$KP \cite{CDZ-JNMP-2017},
and relativistic Toda hierarchy and semi-discrete CLL hierarchy,
both of which are linked to the D$\Delta$mKP \cite{CZZ-SAMP-2021,LZZ-SAMP-2024}.
All these links haven been proved along the line of Ref.\cite{X-IP-1992}
in \cite{C-JMP-2002,CDZ-JNMP-2017,CZZ-SAMP-2021,LZZ-SAMP-2024},
i.e., using recursion operators of the involved $(1+1)$-dimensional integrable hierarchies.

In fact, the master symmetry approach developed by Cheng and Li in \cite{CL-JPA-1992}
reveals many interesting algebraic structures associated with the KP hierarchy.
It was also found in \cite{CL-JPA-1992} that the constraint $u_2=2\varphi_x$
converts \eqref{Lax-pair-t} to the Burgers hierarchy.
To our knowledge, these structures and links  have not yet
been uncovered on the differential-difference level.

In this paper, we will revisit the D$\Delta$KP hierarchy and their constraints.
It is notable that this is not a research simply parallel to Ref.\cite{CL-JPA-1992},
because  there are many new structures in the  differential-difference case
which are different from the continuous case.
For example, the isospectral flows and nonisospectral flows
exhibit new Lie algebraic structures, e.g. \cite{ZC-SAPM-2010,FHTZ-Nonl-2013};
nonisospectral equations allow infinitely many symmetries, e.g.\cite{ZNBC-PLA-2006,ZTOF-JMP-1991};
the  D$\Delta$mKP system allows two different squared eigenfunction symmetry constraints
which give rise to the relativistic Toda hierarchy \cite{CZZ-SAMP-2021}
and semi-discrete CLL hierarchy \cite{LZZ-SAMP-2024}, respectively.
In the present paper we will describe the recursive Lie algebraic structure
in the D$\Delta$KP system based on the master symmetry.
Then we can identity the sdAKNS hierarchy by comparing their  structures.
In a similar way we may show that two new linear eigenfunction constraints  yield
the (combined) semi-discrete Burgers (sdBurgers) hierarchy
from the D$\Delta$KP system and D$\Delta$mKP system, respectively.
Note that this sdBurgers  hierarchy is a combination of the classical sdBurgers hierarchy \cite{Z-PDEAM-2022}
so that it matches to the constraint results.

The paper is organized as follows.
In Section \ref{sec-2},
after introducing some notions and notations as preliminary,
we briefly review the D$\Delta$KP hierarchy, their symmetries,
algebraic structures and especially the recursive structures formulated by its master symmetry.
Then in Section \ref{sec-3}, we prove the squared eigenfunction symmetry constraint of the D$\Delta$KP hierarchy
yields the sdAKNS hierarchy.
Section \ref{sec-4} devotes to the sdBurgers hierarchy, their symmetries and recursive structure.
In Section \ref{sec-5} we show that the combined sdBurgers hierarchy can be obtained
from a new constraint of the D$\Delta$KP system.
In Section \ref{sec-6} we first review the D$\Delta$mKP hierarchy and their symmetries.
Then we show that  a linear eigenfunction constraint leads the D$\Delta$mKP system to
the  same combined sdBurgers hierarchy.
Finally, conclusions  are provided in Section \ref{sec-7}.

\section{D$\Delta$KP hierarchy and algebraic structures}\label{sec-2}	

In this section, we  review the D$\Delta$KP hierarchy, their symmetries,
algebraic structures and especially the recursive structures formulated by the master symmetry.

\subsection{Preliminary}\label{sec-2-1}
				
We briefly list some notions and notations used in this part.
One can also refer to \cite{ZTOF-JMP-1991,FHTZ-Nonl-2013}.
Let $u=u(n, x, \mathbf{t}=(t_1,t_2,\cdots,))$ be a function
of $n\in \mathbb{Z}$ and $x, t_j \in \mathbb{R}$, which is $C^{\infty}$ with respect to $(x,t_j)$,
and goes to zero rapidly as $|n|, |x| \to \infty$.
By $S[u]$ we denote a Schwartz space composed by all functions $f(u)$
that are $C^{\infty}$ differentiable with respect to $u$ and its shifts and derivatives.
Without confusion we write $f(u)=f^{(n)}$.
For two functions $f,g \in S[u]$, the G$\hat{\mathrm{a}}$teaux derivative of $f$ with respect to $u$ in direction $g$
is defined as
\begin{equation}\label{gateaux}
f'[g]\doteq f'(u)[g]=\frac{\mathrm{d}\, f(u+\varepsilon g)}{\mathrm{d}\,\varepsilon}\Big|_{\varepsilon=0},
\end{equation}
using which a Lie product $\llbracket \cdot,\cdot \rrbracket$ is defined as
\begin{equation}\label{lie}
\llbracket f,g\rrbracket=f'[g]-g'[f].
\end{equation}
For an evolution equation
\begin{equation}\label{ut-K}
u_t=K(u)
\end{equation}
where $K(u)\in S[u]$,  we say $\omega=\omega(u)$ is its symmetry
if for all solutions $u$ of \eqref{ut-K} there is
\begin{equation}\label{sym}
\omega_t=K'[\omega].
\end{equation}
For an operator $T$ living on $S[u]$, $T$ is called a strong symmetry  of the equation \eqref{ut-K},
if for any symmetry $\omega$ of  \eqref{ut-K}, the function $T \omega$ is also its symmetry,
which equivalently requires \cite{FF-PD-1981}
\begin{equation}
	T_t=[K', T]\doteq K'T-TK'.
\end{equation}
Here $T_t$ is the total derivative  with respect to $t$, i.e.
$T_t=\frac{\partial T}{\partial t}+T'[u_t]$,
and when $T$ does not contain $t$ explicitly, we have $\frac{\partial T}{\partial t}=0$ and then
\begin{equation}\label{strong-sym}
	T'[K]=[K', T].
\end{equation}
$T$ is hereditary if it satisfies the relation \cite{FF-PD-1981}
\begin{equation}
	T'[T f]g-T'[T g] f =T (T'[f]g-T'[g]f),~~\forall f,g \in S[u].
\end{equation}

By $\Delta$ and $E$ we denote the difference operator and shift operator with respect to $n$,
i.e. $\Delta f^{(n)}=(E-1) f^{(n)}=f^{(n+1)}-f^{(n)}$.
There is an extended Leibniz rule for $\Delta$ \cite{E-book-2005},
\begin{equation}\label{2.7}
\Delta^s g^{(n)}=\sum_{i=0}^{\infty} \mathrm{C}_s^i(\Delta^i g^{(n+s-i)})\Delta^{s-i},~~~ s\in \mathbb{Z},
\end{equation}
where
\begin{equation}\label{C-sj}
\mathrm{C}_0^0=1,~~ \mathrm{C}_s^i=\frac{s(s-1)(s-2)\cdots(s-i+1)}{i!}.
\end{equation}
Note that the above definition indicates that $\mathrm{C}_s^i \equiv 0$, if $i>s>0, ~i,s\in \mathrm{Z}$.
For two functions in $f,g\in S[u]$ we introduce their binary product
\begin{equation}\label{inner-bracket-1}
\langle\!\langle f^{(n)},g^{(n)}\rangle\!\rangle
=\sum_{n=-\infty}^{\infty}\int_{-\infty}^{\infty}f^{(n)} g^{(n)}\, \mathrm{d}x,
\end{equation}
and for an operator $T$  its adjoint operator $T^*$ is defined via
$\langle\!\langle f^{(n)}, Tg^{(n)}\rangle\!\rangle=\langle\!\langle T^* f^{(n)},g^{(n)}\rangle\!\rangle$.

\subsection{Isospectral and nonisopectral D$\Delta$KP flows}\label{sec-2-2}

For the details of this subsection, one can refer to \cite{FHTZ-Nonl-2013}.
To construct isospectral and nonisopectral D$\Delta$KP flows,
we start from the pseudo-difference operator
\begin{equation}\label{L-dKP}
L= \Delta + u + u_1 \Delta^{-1}+ \cdots + u_j \Delta^{-j} + \cdots,
\end{equation}
where $u_j\in S[u]$.
Consider the isospectral Lax triad \cite{FHTZ-Nonl-2013}
\begin{subequations}\label{dkp-iso-lax}
\begin{align}
&L \phi = \eta \phi,~~\eta_{t_m}=0, \label{dkp-iso-lax-L}\\
&\phi_x = A_1 \phi,~~A_1=\Delta+u,\label{dkp-phi_x}\\
&\phi_{t_m} = A_m \phi,~~m=1,2,\cdots, \label{dkp-iso-lax-t}
\end{align}
\end{subequations}
where $A_m=(L^m)_{\geq 0}$ contains the difference part (in terms of $\Delta$)
of $L^m$, and the first three are
\begin{subequations}\label{dkp-Aj}
\begin{align}
&A_1=\Delta + u,\label{A1}\\
&A_2=\Delta^2 + (\Delta u + 2u)\Delta
+(\Delta u + u^2 + \Delta u_1 + 2 u_1),\label{A2}\\
&A_3=\Delta^3 + a_1 \Delta^2 + a_2 \Delta + a_3,
\end{align}
\end{subequations}
where the coefficients $\{a_j\}$ are scalars:
\begin{subequations}
\begin{align*}
a_1 =& \Delta u^2 + 3 \Delta u + 3u,\\
a_2 =& 2 \Delta^2 u + 3 \Delta u + 3u^2 + u \Delta u + \Delta u^2 + 3u_1 + 3 \Delta u_1 + \Delta^2 u_1,\\
a_3 =& \Delta^2 u + u^3 + u\Delta u + \Delta u ^2 + 5u u_1 + (\Delta u )\Delta u_1 + 3u\Delta u_1\\
&+ u_1 \Delta u + u_1 E^{-1} u + 2 \Delta^2 u_1 + 3\Delta u_1 + 3 u_2 + 3\Delta u_2 + \Delta^2 u_2.
\end{align*}
\end{subequations}
 $A_m$ obeys the asymptotic condition
\begin{equation}\label{Am-asy}
	A_m|_{u=0}=\Delta^m.
\end{equation}
The compatibility  of \eqref{dkp-iso-lax} gives rise to
\begin{subequations}\label{dkp-iso-com}
\begin{align}
&L_{x}=[A_1,L],\label{dkp-iso-Lx}\\
&L_{t_m}=[A_m,L],\label{dkp-iso-Ltm}\\
&A_{1,t_m}-A_{m,x}+[A_1,A_m]=0.\label{dkp-iso-zcc}
\end{align}
\end{subequations}
Among them, Eq.\eqref{dkp-iso-Lx} leads to the expression of $u_j~(j >0)$ in terms of $u$, which are
\begin{subequations}
\begin{align}
&u_1 = \Delta^{-1} u_x,\\
&u_2 = \Delta^{-2} u_{xx} - \Delta u_x
- \Delta^{-1}\big(u \Delta^{-1} u_x\big)
+ \Delta^{-1} \big(\big(\Delta^{-1} u_x\big)
E^{-1} u\big),\\
&\cdots.\nonumber
\end{align}
\end{subequations}
Eq.\eqref{dkp-iso-Ltm} determines $A_m$ uniquely under the condition \eqref{Am-asy}.
In other words, if we assume $A_m$ has a form
\[A_m= \Delta^m+  a_1\Delta^{m-1}+ \cdots + a_m,~~
A_m |_{u=0}=\Delta^m,
\]
then $\{a_j\}$ can be uniquely determined from \eqref{dkp-iso-Ltm} and it gives rise to
$A_m=(L^m)_{\geq 0}$.
This fact also indicates the following property.
\begin{proposition}\label{prop-3-1}
For the pseudo-difference operator $L$ given in \eqref{L-dKP}
and any difference operator $\mathcal{A}_m$ with the form:
$\mathcal{A}_m=\sum_{j=0}^m a_j \Delta^{m-j}$,
if they satisfy
\begin{equation}\label{A-L}
[\mathcal{A}_m,L]=0, ~ ~~ \mathcal{A}_m |_{u=0}=0,
\end{equation}
then $\mathcal{A}_m=0$.
\end{proposition}

The third equation \eqref{dkp-iso-zcc} provides the zero curvature representation of
the scalar isopectral D$\Delta$KP hierarchy
\begin{equation}\label{dkp-Km}
u_{t_m} = K_m(u) = A_{m,x} - [A_1, A_m], ~~m=1,2,\cdots.
\end{equation}
The first three equations are
\begin{subequations}\label{dkp-Km-123}
\begin{align}
u_{t_1} = K_1(u)
=& u_x,\label{dkp-k1}\\
u_{t_2} = K_2(u)
=& (1+2\Delta^{-1})u_{xx} - 2 u_x + 2 u u_x,\label{dkp-k2}\\
u_{t_3} = K_3(u)
=& (3\Delta^{-2} + 3\Delta^{-1} +1)u_{xxx} + 3\Delta^{-1} u_x^2 + 3u\Delta^{-1} u_{xx}\nonumber\\
&- 6 \Delta^{-1} u_{xx}
+ 3 u_x + 3 u_x \Delta^{-1} u_x + 3\Delta^{-1}(u u_{xx})\nonumber\\
&+ 3 u u_{xx} - 3 u_{xx} + 3 u_x^2
+ 3 u^2 u_x -6 u u_x.
\end{align}
\end{subequations}
Equation \eqref{dkp-k2} is known as the D$\Delta$KP equation,
which is first presented in \cite{DJM-JPSJ-1982}
and rederived from the pseudo-difference operator $L$ in \cite{KT-CSF-1997}.

To find a master symmetry, we consider the nonisospectral Lax triad \cite{FHTZ-Nonl-2013}
\begin{subequations}\label{dkp-non-lax}
\begin{align}
&L \phi=\eta \phi,~~~\eta_{t_m}=\eta^m+\eta^{m-1},\\
&\phi_x=A_1\phi,\\
&\phi_{t_m}=B_m \phi,~~~m=1,2, \cdots,
\end{align}
\end{subequations}
where $B_m$ is assumed to be a difference operator with the form
\begin{equation}
B_m=\sum_{j=0}^{m} b_j\Delta^{m-j},
\end{equation}
with the asymoptotic condition
\begin{align}\label{dkp-non-bc}
{B_m}|_{u=0}=
x\Delta^m+(x+n)\Delta^{m-1}, ~~~ m \ge 1.
\end{align}
The compatibility of \eqref{dkp-non-lax} is
\begin{subequations}
\begin{align}
& L_x = [A_1, L],\label{dkp-non-Lx}\\
& L_{t_m} = [B_m, L] + L^m + L^{m-1},\label{dkp-non-Ltm}\\
& A_{1,t_m} - B_{m,x} + [A_1, B_m] = 0.\label{dkp-non-zcc}
\end{align}
\end{subequations}
Eq.\eqref{dkp-non-Ltm} determines $B_m$ and first three of them are
\begin{subequations}\label{3.15}
\begin{align}
B_1 =& x A_1 + x + n,\label{B1}\\
B_2 =& x A_2 +(x + n)A_1 + \Delta^{-1} u,\label{B_2},\\
B_3 =& x A_3 +(x+n)A_2 +\Delta^{-1}u\Delta + u\Delta^{-1}u\nonumber\\
&+2\Delta^{-1}u_1 -\Delta^{-1}u +\Delta^{-1}u^2,
\end{align}
\end{subequations}
where $A_j$ are given in \eqref{dkp-Aj}.
Then, \eqref{dkp-non-zcc} defines the  nonisopectral D$\Delta$KP hierarchy
\begin{equation}\label{dkp-sigma-m}
u_{t_m} = \sigma_m(u)
= B_{m,x} - [A_1, B_m],
\end{equation}
which gives
\begin{subequations}\label{dkp-sigma-123}
\begin{align}
u_{t_1} = \sigma_1(u)
=& x K_1 + u,\\
u_{t_2} = \sigma_2(u)
=& x K_2 + (x+n) K_1+ u_x + 3 \Delta^{-1} u_x + u^2 - u,\label{dkp-sigma-2}\\
u_{t_3} = \sigma_3(u)
=& x K_3 + (x+n) K_2 + 5 \Delta^{-2}u_{xx} - 6\Delta^{-1}u_x + 5 \Delta^{-1}(u u_x)\nonumber\\
&+ u_x \Delta^{-1} u + 4u\Delta^{-1}u_x - 2u^2 + u + u^3 + 3u u_x\nonumber\\
&+3\Delta^{-1}u_{xx} + u_{xx} -2u_x,\\
&\cdots,\nonumber
\end{align}
\end{subequations}
where $\{K_m\}$ are defined in \eqref{dkp-Km}.

\subsection{Master symmetry, Lie algebraic structure and recursive structure}\label{sec-2-3}

It is known that the D$\Delta$KP flows $\{K_l\}$ and $\{\sigma_r\}$ forms a Lie algebra with recursive structures, which also leads to symmetries of the isopectral D$\Delta$KP hierarchy.
The algebra of these flows is the following \cite{FHTZ-Nonl-2013}.

\begin{proposition}\label{prop-3-2}
The  flows $\{K_l(u)\}$ and $\{\sigma_r(u)\}$ span a Lie algebra with  structure
\begin{subequations}
\begin{align}
&\llbracket K_l, K_r \rrbracket =0,\\
&\llbracket K_l, \sigma_r \rrbracket = l(K_{l+r-1} + K_{l+r-2}),\label{dkp-Km-sigma-r}\\
&\llbracket \sigma_l, \sigma_r \rrbracket = (l-r) (\sigma_{l+r-1} + \sigma_{l+r-2}),
\end{align}
\end{subequations}
where $l,r \geq 1$ and $K_0(u)$, $\sigma_0(u)$ are set to be $0$.
Eq.\eqref{dkp-Km-sigma-r} indicates that $\sigma_2$ is a master symmetry
which acts as a flow generator through the following relations
\begin{subequations}
\begin{align}
&K_{s+1} = \frac{1}{s} \llbracket K_s, \sigma_2\rrbracket - K_s,~~s\geq1,\label{dkp-master-sym}\\
&\sigma_{s+1} = \frac{1}{s-2} \llbracket \sigma_s, \sigma_2 \rrbracket - \sigma_s,~~s \geq 3,
\end{align}
\end{subequations}
with initial flows $K_1 = u_x$ given in \eqref{dkp-Km-123} and $\sigma_1$ and $\sigma_3$ given in \eqref{dkp-sigma-123}.
\end{proposition}

Apart from the recursive structures for the flows $\{K_l\}$ and $\{\sigma_r\}$,
there   exist  a set of recursive structures for operators $\{A_m\}$ and $\{B_s\}$ as well.
In the following, similar to the continuous case in \cite{CLB-JPA-1988},
we derive such structures for $\{A_m \}$ and $\{B_s\}$ in semi-discrete case.
They will play important roles to identity the connection
between the  D$\Delta$KP system and the sdAKNS hierarchy.

\begin{proposition}\label{prop-3-3}
The operators $\{A_m\}$ and $\{B_s\}$ satisfy the following relation
\begin{subequations}
\begin{align}
&A_{m+1} = \frac{1}{m}
\left(
A_m'[\sigma_2] - B_2'[K_m] + [A_m, B_2]
\right)
-A_m,~~m\geq 1,\label{A_{m+1}-res}\\
&B_{s+1} = \frac{1}{s-2}
\big(
B_s'[\sigma_2] - B_2'[\sigma_s] + [B_s, B_2]
\big) -B_s, ~~~ s \geq 3,\label{B_{n+1}-res}
\end{align}
\end{subequations}
where $A_1$ is defined as \eqref{A1}, $B_1, B_2$ and $B_3$ are given in \eqref{3.15}.
Thus, $B_2$ acts as a ``master operator'' in the above recursive structures.
\end{proposition}

\begin{proof}
Inserting the expression
\begin{equation}\label{A_{m,x}}
A_{m,x} =\partial_x A_m - A_m \partial_x
=[\partial_x, A_m],
\end{equation}
into the zero curvature representation
\eqref{dkp-Km} and \eqref{dkp-sigma-m} respectively, yields
\begin{subequations}\label{3.24}
\begin{align}
&K_m=[A_m,\mathcal{L}], \label{3.24a}\\
&{\sigma}_m=[B_m,\mathcal{L}], \label{3.24b}
\end{align}
\end{subequations}
where
\begin{equation}\label{def-mathcal-L}
\mathcal{L}=A_1-\partial_x.
\end{equation}
The G$\hat{\mathrm{a}}$teaux derivative of $L$ with respect to $u$ reads
\begin{equation}
\mathcal{L}'=(A_1-\partial_x)'=(\Delta+u-\partial_x)'=I,
\end{equation}
where $I$ stands for the identity operator.
Using this fact and the expressions in \eqref{3.24}, we have
\begin{equation}\label{Km'sigma2}
\begin{aligned}
K_m'[\sigma_2]
&=([A_m, \mathcal{L}])'[\sigma_2]
=[A_m'[\sigma_2],\mathcal{L}]
+[A_m,\mathcal{L}'[\sigma_2]]\\
&=[A_m'[\sigma_2],\mathcal{L}]
+[A_m,\sigma_2]
=[A_m'[\sigma_2],\mathcal{L}]
+[A_m,[B_2,\mathcal{L}]],
\end{aligned}
\end{equation}
and
\begin{equation}\label{sigma2'Km}
\begin{aligned}
\sigma'_2[K_m]
&=([B_2, \mathcal{L}])'[K_m]
=[B_2'[K_m],\mathcal{L}]
+[B_2,\mathcal{L}'[K_m]]\\
&=[B_2'[\sigma_2],\mathcal{L}]
+[B_2,K_m]
=[B_2'[\sigma_2],\mathcal{L}]
+[B_2,[A_m,\mathcal{L}]].
\end{aligned}
\end{equation}
Substracting \eqref{sigma2'Km} from \eqref{Km'sigma2} gives
\begin{align}
\llbracket K_m, \sigma_2 \rrbracket
&=K_m'[\sigma_2] - {\sigma_2}'[K_m] \nonumber \\
&=[A_m'[\sigma_2]-B_2'[K_m],\mathcal{L}] +[A_m,[B_2,\mathcal{L}]]-[B_2,[A_m,\mathcal{L}]]  \nonumber \\
&=[A_m'[\sigma_2]-B_2'[K_m],\mathcal{L}] + [[A_m,B_2],\mathcal{L}] \nonumber\\
&=[A_m'[\sigma_2]-B_2'[K_m]+[A_m,B_2],\mathcal{L}], \label{km-sigma2-1}
\end{align}
where  the Jacobi identity of the Lie bracket $[\cdot, \cdot]$
has been used.
On the other hand, from \eqref{dkp-Km-sigma-r} and \eqref{3.24a} we find
\begin{equation}\label{km-sigma2-2}
\llbracket K_m, \sigma_2 \rrbracket
= m(K_{m+1} + K_m)
=[m(A_{m+1}+A_m),\mathcal{L}].
\end{equation}
Meanwhile, direct calculation shows that (see also Eq.(4.24b) in \cite{FHTZ-Nonl-2013})
\begin{equation}\label{3.31}
\big(
A_m'[\sigma_2]-B_2'[K_m]+[A_m,B_2]
\big) \big|_{u=0}
= m (\Delta^{m+1}+\Delta^{m}),
\end{equation}
which coincides with (in light of \eqref{Am-asy})
\begin{equation}\label{3.32}
m(A_{m+1}+A_m)|_{u=0}
= m (\Delta^{m+1}+\Delta^{m}).
\end{equation}
Similar to Proposition \ref{prop-3-1}, one can prove that
for any difference operator $\mathcal{A}_m$ with the form
$\mathcal{A}_m=\sum_{j=0}^m a_j \Delta^{m-j}$,
the equation
\begin{equation}\label{A-LL}
[\mathcal{A}_m, \mathcal{L}]=0, ~ ~~ \mathcal{A}_m |_{u=0}=0
\end{equation}
has only zero solution $\mathcal{A}_m=0$.
Thus, comparing  \eqref{km-sigma2-1} and \eqref{km-sigma2-2}, and noticing the same boundary condition
in \eqref{3.31} and \eqref{3.32},
we conclude  that   \eqref{A_{m+1}-res} holds.
Eq.\eqref{B_{n+1}-res} can be proved in a similar way.

\end{proof}

\section{Constraint of D$\Delta$KP system to sdAKNS hierarchy}\label{sec-3}

The plan of this section is the following.
First, we recall the sdAKNS hierarchy, including the isospectral flows and nonisospectral flows,
and their Lie algebraic structure and master symmetry.
Then we recall the squared eigenfunction symmetry constraint
of the D$\Delta$KP system and present the  constraint result.
Finally, we prove the result along the line of Ref.\cite{CL-JPA-1992},
but we will see that the extensions to the differential-difference case are quite nontrivial.

\subsection{The sdAKNS hierarchy}\label{sec-3-1}

Let $U= (Q, R)^T$ where $Q\doteq Q^{(n)}, R\doteq R^{(n)}$ are scalar  functions of
$n\in \mathbb{Z}$ and $x, t_j \in \mathbb{R}$, which are $C^{\infty}$ with respect to $(x,t_j)$,
and tend to zero rapidly as $|n|, |x| \to \infty$.
Let $\mathcal{V}_2[U]$ denote the Schwartz space composed by all 2-dimensional vector fields
$f(U)=(f_1(U), f_2(U))^T$
that are $C^{\infty}$ differentiable with respect to $U$ and its shifts and derivatives.
We can extend the notions described in Sec.\ref{sec-2-1} to the case of 2-dimensional vector fields.
For example, for two functions $f,g \in \mathcal{V}_2[U]$, the G$\hat{\mathrm{a}}$teaux derivative
of $f$ with respect to $U$ in direction $g$
is defined as
\begin{equation}\label{gateaux-U}
f'[g]\doteq f'(U)[g]=\frac{\mathrm{d}\, f(U+\varepsilon g)}{\mathrm{d}\,\varepsilon}\Big|_{\varepsilon=0}.
\end{equation}
In addition, we denote
\begin{equation}\label{inner-bracket}
\langle f,g \rangle=f_1g_1+f_2 g_2,~~~ f,g \in \mathcal{V}_2[U].
\end{equation}

Consider the following spectral problem
\begin{equation}\label{sp-DT}
\Theta^{(n+1)} = \left(
                    \begin{array}{cc}
                      \lambda  + Q  R  & Q  \\
                      R  & 1
                    \end{array}
                  \right) \Theta^{(n)},~~~
\Theta^{(n)}=(\theta_{1}^{(n)}, \theta_{2}^{(n)})^T,
\end{equation}
where $ \lambda$ is the spectral parameter.
This spectral problem was first presented by Date, Jimbo and Miwa in \cite{DJM-JPSJ-1983} in 1983,
and later was studied by Tu and his collaborators \cite{ZTOF-JMP-1991,MRT-IP-1994}
for the integrable characteristics of the sdAKNS hierarchy.
It also acts as a Darboux  transformation of the AKNS spectral problem
after replacing $R^{(n)}$ with $R^{(n+1)}$ and $\lambda$ with $2(\lambda-\beta)$
where $R^{(n+1)}$  represents a new solution generated from $(Q^{(n)},R^{(n)})$
by introducing a soliton parameter $\beta$ \cite{AY-JPA-1994}.
On the other hand, the Darboux  transformation can be considered as a
discrete spectral problem to generate the fully discrete AKNS system and
fully discrete KP equation \cite{CXZ-proc-2020}.

For the  sdAKNS hierarchy derived from the spectral problem \eqref{sp-DT},
we employ the formulae presented in \cite{CDZ-JNMP-2017}.
The isospectral and nonisospectral sdAKNS hierarchies are
\begin{align}
&U_{t_s}=
\mathcal{K}_s
= \mathfrak{L}^s  \mathcal{K}_0,~~
\mathcal{K}_0
=\left(\begin{array}{c}
Q \\-R
\end{array}\right),~~s=0,1,2,\cdots , \label{sdakns-kn}
\end{align}
and
\begin{align}
&U_{t_s}
={\gamma}_s
= \mathfrak{L}^s  \gamma_0,~~
\gamma_0
=\left(\begin{array}{c}
(n+\frac{1}{2})Q \\
-(n-\frac{1}{2})R
\end{array}\right),
~~s=0,1,2,\cdots,
\end{align}
in which $\mathfrak{L}$ is the recursion operator
\begin{equation}\label{L-akns}
\mathfrak{L} =
\left(\begin{array}{cc}
1 & 0\\
0 & E^{-1}
\end{array}\right)
\left[
\mu^{(n)} I
-\left(\begin{array}{c}
Q^{(n)}\\-R^{(n+1)}
\end{array}\right)
(E+1)\Delta^{-1}
\left(R^{(n+1)}, Q^{(n)}\right)
\right]
\left(\begin{array}{cc}
E & 0\\
0 & 1
\end{array}\right)
{-I},
\end{equation}
with $\mu^{(n)} = 1 + Q^{(n)} R^{(n+1)}$
and $I$ being the identity matrix of the second order.
The first two nonlinear isospectral  flows are
\begin{align}
\mathcal{K}_1
=\left(\begin{array}{c}
\mathcal{K}_{1,1}\\
\mathcal{K}_{1,2}
\end{array}\right)
=\left(\begin{array}{c}
Q^{(n+1)} - Q^{(n)} - {(Q^{(n)})}^2 R^{(n)} \\
-R^{(n-1)} + R^{(n)} + {(R^{(n)})}^2 Q^{(n)}
\end{array}\right),
\end{align}
\begin{subequations}\label{K2}
\begin{align}
\mathcal{K}_2
=\left(\begin{array}{c}
\mathcal{K}_{2,1}\\
\mathcal{K}_{2,2}
\end{array}\right),
\end{align}
where
\begin{align}
\mathcal{K}_{2,1}&=
Q^{(n+2)} - (Q^{(n+1)})^2 R^{(n+1)}
- (Q^{(n)})^2 R^{(n-1)} - 2 Q^{(n+1)} Q^{(n)} R^{(n)} \nonumber \\
&~~~~~ + (Q^{(n)})^3 (R^{(n)})^2 -2 Q^{(n+1)}
+2 (Q^{(n)})^2 R^{(n)} + Q^{(n)},\\
\mathcal{K}_{2,2}&=
-R^{(n-2)} + {(R^{(n-1)})}^2 Q^{(n-1)}
+ {(R^{(n)})}^2 Q^{(n+1)}
+ 2 R^{(n-1)} R^{(n)} Q^{(n)} \nonumber\\
&~~~~~
- {(R^{(n)})}^3 {(Q^{(n)})}^2 + 2 R^{(n-1)}
- 2 {(R^{(n)})}^2 Q^{(n)} - R^{(n)},
\end{align}
\end{subequations}
and the first nonlinear  nonisospectral flow is
\begin{align}
\gamma_1
=\left(\begin{array}{c}
(n + \frac{3}{2}) Q^{(n+1)}
- (n + \frac{3}{2}) (Q^{(n)})^2 R^{(n)}
- 2 Q^{(n)} \Delta^{-1} (Q^{(n)} R^{(n)})
- (n + \frac{1}{2}) Q^{(n)}\\
- (n  {-} \frac{3}{2}) R^{(n-1)}
+ (n + \frac{1}{2}) {(R^{(n)})}^2 Q^{(n)}
+ 2 R^{(n)} \Delta^{-1} (R^{(n)} Q^{(n)})
+ (n - \frac{1}{2}) R^{(n)}
\end{array}\right).\label{sdakns-gamma-1}
\end{align}

These flows compose a Lie algebra by which the recursive relations are described.

\begin{proposition}\label{prop-4-1}\cite{CDZ-JNMP-2017}
The flows $\{\mathcal{K}_s,\gamma_l\}_{s \geq 0, l \geq 0}$
compose a Lie algebra with the following structure
\begin{subequations}
\begin{align}
&\llbracket
\mathcal{K}_s, \mathcal{K}_l
\rrbracket = 0,\\
&\llbracket
\mathcal{K}_s, \gamma_l
\rrbracket = s(\mathcal{K}_{s+l} + \mathcal{K}_{s+l-1}),\label{sdakns-ks-gamma-l} \\
&\llbracket
\gamma_s, \gamma_l
\rrbracket = (s-l) (\gamma_{s+l} + \gamma_{s+l-1}),
\end{align}
\end{subequations}
where we define $\mathcal{K}_{-1} = \gamma_{-1} =0 $.
Such a structure  indicates $\gamma_1$ is the master symmetry by which
the flows can be generated recursively:
\begin{subequations}
\begin{align}
\mathcal{K}_{s+1} = \frac{1}{s} \llbracket \mathcal{K}_s, \gamma_1 \rrbracket - \mathcal{K}_{s},~~s\geq 1,
\label{3.11a} \\
\gamma_{s+1} = \frac{1}{s-1} \llbracket \gamma_s, \gamma_1 \rrbracket - \gamma_{s}, ~~s\ge 2.
\end{align}
\end{subequations}
\end{proposition}

Note that the above sdAKNS hierarchy obtained from \eqref{sp-DT}
are apparently different from the Ablowitz-Ladik hierarchy which admits many reductions
\cite{AL-JMP-1976,ZC-SAPM-2010}.
For the isospectral sdAKNS hierarchy \eqref{sdakns-kn},
so far only one reduction was found in \cite{CZZ-SAMP-2021}, which yields the Volterra hierarchy.

\subsection{Squared eigenfunction symmetry, constraint and sdAKNS hierarchy}\label{sec-3-2}

\subsubsection{Main result}\label{sec-3-2-1}

For the Lax triad \eqref{dkp-iso-lax} where $\phi$ is the eigenfunction,
consider the following adjoint form:
\begin{subequations}\label{dkp-iso-adj-lax}
\begin{align}
&-\phi_x^* =   A_1^* \phi^*,~~A_1^*= \Delta^*+u,\label{dkp-adj-phi_x}\\
&-\phi^*_{t_m} = A_m^* \phi^*,~~m=1,2,\cdots, \label{dkp-adj-phi_t}
\end{align}
\end{subequations}
where $\phi^*$ is a solution of this system, $A^*_m$ is the adjoint form of $A_m$,
which is defined with respect to the binary form \eqref{inner-bracket-1}.
Then, it can be proved that
$(\phi\phi^*)_x$ is a symmetry of the isospectral D$\Delta$KP hierarchy \eqref{dkp-Km},
provided $\phi$ and $\phi^*$ satisfy  \eqref{dkp-iso-lax} and \eqref{dkp-iso-adj-lax}
\cite{LCTLH-MM-2016,CDZ-JNMP-2017,CZZ-SAMP-2021}.
Note that such a symmetry is conventionally called a ``squared eigenfunction symmetry''
because in the self-adjoint case (e.g. for the KdV hierarchy),
the symmetry reads $\frac{1}{2}(\phi^2)_x$ due to $\phi^*=\phi$.

Since $u_x$ is a symmetry of the D$\Delta$KP hierarchy \eqref{dkp-Km} as well,
one can consider the constraint $u_x-(\phi\phi^*)_x=0$,
which leads to the sdAKNS hierarchy.

\begin{theorem}\label{Th-3-1}
Consider the constraint
\begin{equation}\label{u=-qr}
u  = - Q  R,
\end{equation}
where we have taken $\phi=Q$ and $\phi^*=-R$.
Under this constraint,
the coupled system composed by \eqref{dkp-phi_x} and \eqref{dkp-adj-phi_x} yields equations
\begin{subequations}\label{K1}
\begin{align}
&Q^{(n)}_{x}
= Q^{(n+1)} - Q^{(n)} - (Q^{(n)})^2 R^{(n)}, \label{3.14a}\\
&R^{(n)}_x
=  - R^{(n-1)} + R^{(n)}  + (R^{(n)})^2 Q^{(n)},
\end{align}
\end{subequations}
which provides the $t_1$-flow $\mathcal{K}_1$ in the sdAKNS hierarchy \eqref{sdakns-kn}.
Under this constraint \eqref{u=-qr} together with \eqref{K1},
the D$\Delta$KP spectral problem \eqref{dkp-iso-lax-L} gives rise to the sdAKNS spectral problem \eqref{sp-DT}
(after gauge transformations),
and \eqref{dkp-iso-lax-t} and \eqref{dkp-adj-phi_t}
together compose the isospectral sdAKNS hierarchy \eqref{sdakns-kn}.
\end{theorem}

This theorem has been proved in \cite{CDZ-JNMP-2017} where \eqref{dkp-iso-lax-t} and \eqref{dkp-adj-phi_t}
are converted into
\begin{equation}
\left(\begin{array}{c}
Q \\
R
\end{array}\right)_{t_{m+1}}
=\mathfrak{L}
\left(\begin{array}{c}
Q \\
R
\end{array}\right)_{t_{m}},~~~ m=1,2,\cdots,
\end{equation}
with $\mathfrak{L}$ given by \eqref{L-akns}
and the first equation $(Q,R)^T_{t_1}=(Q,R)^T_{x}=\mathcal{K}_1$ given by \eqref{K1}.
In the next, we will prove these results along the line of \cite{CL-JPA-1992},
by making use of the master symmetries, Lie algebraic and recursive structures
in the  D$\Delta$KP system and sdAKNS hierarchy.

\subsubsection{Proof of Theorem \ref{Th-3-1}}\label{sec-3-2-2}

The sdAKNS spectral problem \eqref{sp-DT} can be obtained from \eqref{dkp-iso-lax-L}
 but some gauge transformations are needed. One can refer to \cite{CDZ-JNMP-2017} for details.
In the following we only prove that  \eqref{dkp-iso-lax-t} and \eqref{dkp-adj-phi_t}
give rise to the isospectral sdAKNS hierarchy \eqref{sdakns-kn}.
The proof will be completed by proving a series of lemmas.
We employ the following notation used in \cite{CL-JPA-1992} to denote
\begin{equation}
F(u)|_R = F(u)|_{\eqref{u=-qr},\eqref{K1}} ,
\end{equation}
which means
we have replaced all $\partial_x^j Q^{(n+i)}$ and $\partial_x^j R^{(n+i)}$ in $F(u)$
so that $F(u)|_{R}$ has an expression  in terms of only $U$ and its shifts $U^{(n+i)}$.
These lemmas are the following.

\begin{lemma}\label{lem-sdakns-1}
For the isospectral D$\Delta$KP flows $\{K_m(u)\}$ and master symmetry $\sigma_2(u)$
given in \eqref{dkp-Km} and \eqref{dkp-sigma-2},
and for the isospectral sdAKNS flows $\{\mathcal{K}_m(U)\}$
and the master symmetry $\gamma_1(U)$ given in \eqref{sdakns-kn} and \eqref{sdakns-gamma-1},
we have  the following:
\begin{subequations}
\begin{align}
& K_m(u) |_R = -\langle U, J \mathcal{K}_m(U ) \rangle,\label{kn-r-akns}\\
& \sigma_2(u) |_R = -\langle U, J \ti{\gamma}_1(U ) \rangle,\label{sigma-r-akns}
\end{align}
\end{subequations}
where
\begin{equation}
\ti{\gamma}_1(U)
= x \mathcal{K}_2 + x \mathcal{K}_1 - \frac{1}{2} \mathcal{K}_1 + \gamma_1,\label{ti-gamma-1}
\end{equation}
$J = \bigl(\begin{smallmatrix}
0 & 1\\
1 & 0
\end{smallmatrix}\bigr)$
and $\langle \cdot, \cdot \rangle$ is the inner product defined in \eqref{inner-bracket}.
\end{lemma}

\begin{proof}
Using \eqref{K1} we immediately have
\begin{align}
K_1(u) |_R &= u_x|_R
= (-Q R)_x=-Q_x R-Q R_x \nonumber\\
&= -( Q^{(n+1)} - Q^{(n)} - (Q^{(n)})^2 R^{(n)})R^{(n)}
  -Q^{(n)}(- R^{(n-1)} + R^{(n)}  + (R^{(n)})^2 Q^{(n)}) \nonumber\\
&= - (Q^{(n)}, R^{(n)})
\left(\begin{array}{cc}
0 & 1\\
1 & 0
\end{array}\right)
\left(\begin{array}{c}
Q^{(n+1)} - Q^{(n)} - (Q^{(n)})^2 R^{(n)}\\
- R^{(n-1)} + R^{(n)}  + (R^{(n)})^2 Q^{(n)}
\end{array}\right) \nonumber \\
&= - \langle U  , J \mathcal{K}_1(U )\rangle. \label{3.19}
\end{align}
For $K_2 (u)$ defined in \eqref{dkp-k2} and $\mathcal{K}_2(U )$ defined in \eqref{K2},
using \eqref{K1}
we can also verify that
\begin{equation}\label{3.20}
K_2 (u) |_R = -\langle U, J \mathcal{K}_2(U ) \rangle.
\end{equation}
The proof of \eqref{sigma-r-akns} is long but straight forward
by using \eqref{3.19} and \eqref{3.20}, i.e. \eqref{kn-r-akns} for $m=1,2$.
Here we skip presenting the verification details.

In the next we prove \eqref{kn-r-akns} is true for general $m$ by means of mathematical induction.
We assume the following relation holds:
\begin{equation}\label{3.19-m}
K_m(u) |_R = -\langle U, J \mathcal{K}_m(U ) \rangle
\end{equation}
and go to prove
\begin{equation}\label{3.19-m+1}
K_{m+1}(u) |_R = -\langle U, J \mathcal{K}_{m+1}(U ) \rangle.
\end{equation}
Note that \eqref{3.19-m} indicates
\begin{equation}\label{3.19-mf}
(K_m(u) |_R)'(U)[f] = -\langle f, J \mathcal{K}_m(U ) \rangle
-\langle U, J  \mathcal{K}_m^{\,\prime}(U )[f] \rangle, ~~~  f\in \mathcal{V}_2[U],
\end{equation}
where the G$\hat{\mathrm{a}}$teaux derivatives are taken with respect to $U$ in direction $f$.
In addition, denoting $\ti{\gamma}_1=(\ti{\gamma}_{1,1}, \ti{\gamma}_{1,2})^T$,
then \eqref{sigma-r-akns} is written as
\begin{equation}\label{sigma2-QR}
\sigma_2(u)|_R
= -( {Q} \ti{\gamma}_{1,2}
+  {R} \ti{\gamma}_{1,1}).
\end{equation}
With all these in hand, we have
\begin{align*}
K_m^{\,\prime}(u)[\sigma_2]|_R
&=
\Bigl(
\frac{\mathrm{d}}{\mathrm{d}\,\varepsilon}
{K_m( u +\varepsilon \sigma_2) }
\big|_{\varepsilon=0}\Bigr)\Big|_{R}
\\
&=\frac{\mathrm{d}}{\mathrm{d}\,\varepsilon}
{K_m( - Q R +\varepsilon (\sigma_2 |_R)) }\big|_{\varepsilon=0}\\
&= \frac{\mathrm{d}}{\mathrm{d}\,\varepsilon}
{K_m( - Q R -\varepsilon Q \ti{\gamma}_{1,2}
-\varepsilon R  \ti{\gamma}_{1,1})} \big|_{\varepsilon=0}\\
&= \frac{\mathrm{d}}{\mathrm{d}\,\varepsilon}
{K_m \bigl(
- (Q  + \varepsilon \ti{\gamma}_{1,1})
(R  + \varepsilon \ti{\gamma}_{1,2})\bigr)} \big|_{\varepsilon=0}\\
&=(K_m(u) |_R)'(U)[\ti{\gamma}_1].
\end{align*}
Then, using \eqref{3.19-mf}, we immediately obtain
\begin{equation}\label{sdakns-km'sigma}
K_m^{\,\prime}(u)[\sigma_2]|_R= -\langle
\ti{\gamma}_1,
J \mathcal{K}_m(U )
\rangle
-\langle
U,
J  \mathcal{K}_m^{\,\prime}(U) [\ti{\gamma}_1]
\rangle.
\end{equation}
Similarly we have
\begin{equation}\label{sdakns-sigma'km}
\sigma_2^{\,\prime}(u)[K_m]|_R
= -\langle
\mathcal{K}_m,
J \ti{\gamma}_1(U)
\rangle
-\langle
U,
J  \ti{\gamma}_1^{\,\prime}(U)[\mathcal{K}_m]
\rangle.
\end{equation}
Subtracting \eqref{sdakns-sigma'km} from \eqref{sdakns-km'sigma} yields
\begin{align*}
\llbracket
K_m, \sigma_2
\rrbracket|_R
&= K_m^{\,\prime}(u)[\sigma_2]|_R
-  \sigma_2^{\,\prime}(u)[K_m]|_R\\
&= - \langle
U,
J (\llbracket\mathcal{K}_m, \ti{\gamma}_1\rrbracket)
\rangle
-\langle
\ti{\gamma}_1,
J \mathcal{K}_m(U)
\rangle
+ \langle
\mathcal{K}_m,
J \ti{\gamma}_1(U)
\rangle\\
&= - \langle
U,
J (\llbracket\mathcal{K}_m,\ti{\gamma}_1\rrbracket)
\rangle.
\end{align*}
From the definition \eqref{ti-gamma-1} for $\ti{\gamma}_1$, and also in light of \eqref{3.11a},
we know that
\begin{equation}\label{sdakns-km-gamma-1}
\llbracket \mathcal{K}_m, \ti{\gamma}_1\rrbracket
=\llbracket \mathcal{K}_m,  {\gamma}_1\rrbracket
= m(\mathcal{K}_{m+1}+\mathcal{K}_m), ~~ m\geq1.
\end{equation}
It then follows that
\begin{equation}\label{sdakns-sim-km-sigma2}
\llbracket
K_m, \sigma_2
\rrbracket|_R
= - m \langle
U,
J (\mathcal{K}_{m+1} + \mathcal{K}_m)
\rangle.
\end{equation}
Thus, using \eqref{dkp-master-sym}   we have
\begin{align*}
K_{m+1}|_R
&= \frac{1}{m}
\llbracket K_m, \sigma_2 \rrbracket \big|_R
- K_m \big|_R\\
&= - \langle
U,
J (\mathcal{K}_{m+1} + \mathcal{K}_m)
\rangle
+ \langle
U,
J \mathcal{K}_m
\rangle\\
&= -\langle U, J \mathcal{K}_{m+1}(U)\rangle,
\end{align*}
which comploetes the proof for \eqref{3.19-m+1}.

\end{proof}

\begin{lemma}\label{lem-3.2}
For the D$\Delta$KP operators $A_m$ and $B_2$,
flows $\{K_m\}$ and master symmetry $\sigma_2$,
and the sdAKNS flows $\{\mathcal{K}_m\}$ and master symmetry $\gamma_1$,
the following formulae hold:
\begin{subequations}
\begin{align}
&A_m^{\,\prime}(u)[\sigma_2(u)] |_R
= (A_m |_R)' (U)[\ti{\gamma}_1(U)],\label{sdakns-Am'sigma2}\\
&B_2^{\,\prime}(u) [K_m(u)] |_R
= (B_2 |_R )'(U)[\mathcal{K}_m (U)].\label{sdakns-B2'km}
\end{align}
\end{subequations}
\end{lemma}

\begin{proof}
First, we have
\begin{align*}
 A_m^{\,\prime}(u)[\sigma_2] |_R
&= \frac{\mathrm{d}}{\mathrm{d}\,\varepsilon}
{A_m( u +\varepsilon \sigma_2) }\big|_{u=-Q R,\,\varepsilon=0}\\
&= \frac{\mathrm{d}}{\mathrm{d}\,\varepsilon}
{A_m( - Q R  -\varepsilon Q  \ti{\gamma}_{1,2}
-\varepsilon R  \ti{\gamma}_{1,1})}\big|_{\varepsilon=0},
\end{align*}
where we have used \eqref{sigma-r-akns} in Lemma \ref{lem-sdakns-1}.
From the above expression, we then have
\begin{align*}
A_m^{\,\prime}(u)[\sigma_2] |_R
&= \frac{\mathrm{d}}{\mathrm{d}\,\varepsilon}
{ A_m (
- (Q + \varepsilon \ti{\gamma}_{1,1})
(R + \varepsilon \ti{\gamma}_{1,2}))} \big|_{\varepsilon=0}\\
&=
\frac{\mathrm{d}}{\mathrm{d}\,\varepsilon}
{ (A_m |_R) (U  + \varepsilon \ti{\gamma}_1) } \big|_{\varepsilon=0}\\
&= (A_m |_R)'(U)[\ti{\gamma}_1],
\end{align*}
which yields \eqref{sdakns-Am'sigma2}.
In a similar way, using \eqref{kn-r-akns} and writing $\mathcal{K}_m=(\mathcal{K}_{m,1},\mathcal{K}_{m,2})^T$,
we have
\begin{align*}
 B_2^{\,\prime}(u) [K_m] |_R
&= \frac{\mathrm{d}}{\mathrm{d}\,\varepsilon}
{B_2( u +\varepsilon K_m) }\big|_{u=-QR,\,\varepsilon=0}\\
&=
\frac{\mathrm{d}}{\mathrm{d}\,\varepsilon}
{B_2( - Q R
-\varepsilon Q  \mathcal{K}_{m,2}
-\varepsilon R  \mathcal{K}_{m,1}
)}
\big|_{\varepsilon=0}\\
&= \frac{\mathrm{d}}{\mathrm{d}\,\varepsilon}
{ B_2 (
- (Q + \varepsilon \mathcal{K}_{m,1})
(R  + \varepsilon \mathcal{K}_{m,2}))} \big|_{\varepsilon=0}\\
&=
\frac{\mathrm{d}}{\mathrm{d}\,\varepsilon}
{ (B_2 |_R) (U  + \varepsilon \mathcal{K}_m) } \big|_{\varepsilon=0}\\
&= (B_2 |_R)'(U )[\mathcal{K}_m],
\end{align*}
which yields \eqref{sdakns-B2'km}.

\end{proof}

\begin{lemma}\label{lem-3.3}
Define flows $\{\zeta_m(U)\}$ by
\begin{subequations}
\begin{equation}
\zeta_1 (U)
=\left(\begin{array}{c}
-Q^{(n+1)} + (Q^{(n)})^2 R^{(n)} + Q^{(n)} \Delta^{-1} Q^{(n)} R^{(n)}\\
-R^{(n-1)} + R^{(n)} - R^{(n)} \Delta^{-1} Q^{(n)} R^{(n)}
\end{array}\right),
\end{equation}
and
\begin{equation}
\zeta_{m+1} (U)
=  \zeta_m^{\,\prime}(U) [\ti{\gamma}_1(U)]
- \left(\begin{array}{cc}
B_2 & 0\\
0 & -B_2^*
\end{array}\right)\Bigg|_R
\zeta_m(U),~~m  {\ge} 1.
\label{3.21b}
\end{equation}
\end{subequations}
where $B_2$ is the operator given in \eqref{B_2} and $\ti{\gamma}_1(U)$ is defined in \eqref{ti-gamma-1}.
Then, for the operator $\mathcal{L}$ defined in \eqref{def-mathcal-L} and operator $A_s$, we have
\begin{subequations}\label{L,An-zeta}
\begin{align}
&\left(\begin{array}{cc}
\mathcal{L} & 0\\
0 & -\mathcal{L}^*
\end{array}\right)\Bigg|_R
\zeta_m
= 0,\label{L-zeta}\\
&\left(\begin{array}{cc}
A_s & 0\\
0 & -A_s^*
\end{array}\right)\Bigg|_R \zeta_m
= \zeta_m^{\,\prime}(U)[\mathcal{K}_s].\label{An-zeta}
\end{align}
\end{subequations}
\end{lemma}

\begin{proof}
Denote $\zeta_m=(\zeta_{m,1}, \zeta_{m,2})^T$.
Let us  prove the first component of \eqref{L-zeta},
for which we need to check
\begin{equation*}
\mathcal{L}|_R \,\zeta_{1,1}
=(A_1 - \partial_x)|_R \,\zeta_{1,1}
=(A_1 - \partial_x)|_R
\Big(-Q^{({n+1})} + (Q^{(n)})^2 R^{(n)} + Q^{(n)} \Delta^{-1} Q^{(n)} R^{(n)}\Big).
\end{equation*}
It vanishes after inserting \eqref{K1}. The verification is straight forward.
Then we assume the first component of \eqref{L-zeta} is satisfied for $\zeta_{m,1}$, i.e.
\begin{equation}\label{m}
\mathcal{L}|_R \, \zeta_{m,1}=0.
\end{equation}
For $\zeta_{m+1,1}$ we have
\begin{align*}
\mathcal{L}|_R \,\zeta_{m+1,1}
&= \mathcal{L}|_R \big(  \zeta_{m,1}^{\,\prime}(U)[\ti{\gamma}_1]
- B_2|_R \,\zeta_{m,1} \big)\\
&= \mathcal{L}|_R \zeta_{m,1}^{\,\prime}(U)[\ti{\gamma}_1]
- \mathcal{L}|_R B_2|_R \, \zeta_{m,1}
+ B_2|_R \mathcal{L}|_R \, \zeta_{m,1}\\
&= \mathcal{L}|_R \zeta_{m,1}^{\,\prime}(U) [\ti{\gamma}_1]
+ [B_2, \mathcal{L}]|_R \, \zeta_{m,1},
\end{align*}
where we have used \eqref{m}.
Meanwhile, we have
\begin{align*}
(\mathcal{L}|_R)'(U)[\ti{\gamma}_1]
&=
 (\Delta - Q R -\partial_x)'(U)[\ti{\gamma_1}]
= - (Q R)'[
\ti{\gamma_1}
]
= - R  \ti{\gamma}_{1,1}
- Q  \ti{\gamma}_{1,2}
= - \langle  {U}, J\ti{\gamma}_1 \rangle\\
&= \sigma_2 |_R
= [B_2, \mathcal{L}] |_R,
\end{align*}
where we have successively used \eqref{sigma-r-akns} and \eqref{3.24b}.
It then follows that
\begin{align*}
\mathcal{L}|_R \,\zeta_{m+1,1}
= \mathcal{L}|_R
 \zeta_{m,1}^{\,\prime}(U) [\ti{\gamma}_1]
+ (\mathcal{L}|_R)'(U)[\ti{\gamma}_1] \zeta_{m,1}
= (\mathcal{L}|_R \, \zeta_{m,1})'(U)[\ti{\gamma}_1]
= 0,
\end{align*}
where, again, use has been made of \eqref{m}.
Thus, we have proved the first component of  \eqref{L-zeta}. The second component can be proved in a similar way.
	
In the next, we are going to prove \eqref{An-zeta}.
First, for arbitrary $m \geq 1$, from \eqref{L-zeta} we have
\begin{align*}
0 & = \mathcal{L}|_{R}\, \zeta_{m,1}
  =  A_1 |_R \, \zeta_{m,1}- (\zeta_{m,1})_x
  =  A_1 |_R \, \zeta_{m,1}-  \zeta_{m,1}^{\, \prime}(U)[U_x]
  =  A_1 |_R \, \zeta_{m,1}-  \zeta_{m,1}^{\, \prime}(U)[\mathcal{K}_1],
\end{align*}
where we have used the fact $U_x=U_{t_1}=\mathcal{K}_1$.
Thus we have
\begin{equation}
A_1 |_R \, \zeta_{m,1}=\zeta_{m,1}^{\, \prime}(U)[\mathcal{K}_1],
~~m \geq 1.
\end{equation}
We prove the general case  by induction.
Let us assume the following is valid for arbitrary $m$:
\begin{equation}\label{As}
A_s |_R \, \zeta_{m,1}=\zeta_{m,1}^{\, \prime}(U)[\mathcal{K}_s],
~~m \geq 1.
\end{equation}
Then we look for the same relation  for $A_{s+1}$ and $\mathcal{K}_{s+1}$.
To achieve that, we start from the above relation \eqref{As} to calculate
\begin{equation*}
\begin{aligned}
(\zeta_{m,1}^{\, \prime}(U)[\mathcal{K}_s])'(U)[\ti{\gamma}_1]
&= (A_s|_R \,\zeta_{m,1})'(U)[\ti{\gamma}_1]\\
&= (A_s|_R)'(U)[\ti{\gamma}_1] \zeta_{m,1}
+ A_s|_R\, \zeta_{m,1}^{\, \prime}(U)[\ti{\gamma}_1]\\
&= A_s^{\,\prime}(u)[\sigma_2]|_R \, \zeta_{m,1} + A_s|_R\, \zeta_{m,1}^{\, \prime}(U)[\ti{\gamma}_1],
\end{aligned}
\end{equation*}
where relation  \eqref{sdakns-Am'sigma2} has been used.
Then, from the definition \eqref{3.21b} we have
\begin{align}
(\zeta_{m,1}^{\, \prime}(U)[\mathcal{K}_s])'(U)[\ti{\gamma}_1]
&= A_s^{\,\prime}(u)[\sigma_2]|_R \, \zeta_{m,1}
+ A_s |_R B_2 |_R \, \zeta_{m,1}
- A_s |_R B_2 |_R \,  \zeta_{m,1}
+ A_s|_R \,\zeta_{m,1}^{\,\prime}(U)[\ti{\gamma}_1] \nonumber\\
&= (A_s^{\,\prime}(u)[\sigma_2] + A_s B_2)|_R \, \zeta_{m,1}
+ A_s |_R \, \zeta_{m+1,1}.
\label{3.26}
\end{align}
Similarly, we get
\begin{equation}
(\zeta_{m,1}^{\,\prime}(U) [\ti{\gamma}_1])'(U)[\mathcal{K}_s]
= (B_2^{\,\prime}(u)[K_s] +  B_2 A_s)|_R \, \zeta_{m,1}
+  \zeta_{m+1,1}^{\,\prime}(U)[\mathcal{K}_s].
\label{3.27}
\end{equation}
Note that from the assumption \eqref{As} we have
\begin{equation*}
A_s |_R \, \zeta_{m+1,1}
=\zeta_{m+1,1}^{\,\prime}(U)[\mathcal{K}_s].
\end{equation*}
Meanwhile, noticing the relation
\[(f'[g])'[h]-(f'[h])'[g]=f'[\llbracket g, h \rrbracket],~~ f,g,h\in \mathcal{V}_2[U] \]
and the relation (in light of \eqref{3.11a} and \eqref{ti-gamma-1}):
\begin{align}\label{K-ga}
\mathcal{K}_{s+1} = \frac{1}{s} \llbracket \mathcal{K}_s, \gamma_1 \rrbracket - \mathcal{K}_{s}
= \frac{1}{s} \llbracket \mathcal{K}_s, \ti\gamma_1 \rrbracket - \mathcal{K}_{s},
\end{align}
from \eqref{3.26} and \eqref{3.27} we have
\begin{align*}
\zeta_{m,1}^{\, \prime}(U) [\mathcal{K}_{s+1}]
&= \zeta_{m,1}^{\, \prime}(U) [\frac{1}{s} \llbracket \mathcal{K}_s, \ti\gamma_1 \rrbracket - \mathcal{K}_{s}]\\
&=\frac{1}{s}\big(
(\zeta_{m,1}^{\, \prime}(U)[\mathcal{K}_s])'(U)[\ti{\gamma}_1]
-(\zeta_{m,1}^{\, \prime}(U) [\ti{\gamma}_1])'(U)[\mathcal{K}_s]
\big)
-\zeta_{m,1}^{\, \prime}(U)[\mathcal{K}_{s}]\\
&= \frac{1}{s}\big (
A_s^{\, \prime}(u)[\sigma_2] - B_2^{\, \prime}(u)[K_s] +  [A_s, B_2]\big ) \big|_R \, \zeta_{m,1}- A_s |_R \, \zeta_{m,1}\\
&= A_{s+1} |_R \, \zeta_{m,1},
\end{align*}
where we have made use of \eqref{As} and \eqref{A_{m+1}-res}.
Thus, the first component of \eqref{An-zeta} is proved.
The second component can be proved in a similar manner.

\end{proof}

\begin{lemma}\label{thm-dkp-sdakns}
Under the constraint \eqref{u=-qr}, we have the isospectral sdAKNS flows
\begin{equation}
\left(\begin{array}{l}
A_m (u) |_R \, \phi\\
A_m^* (u) |_R \, \phi^*
\end{array}\right)
=\mathcal{K}_m (U), ~~~ m=1,2,\cdots. \label{sdakns-Am-r}
\end{equation}
In addition, we have
\begin{equation}
\left(\begin{array}{l}
B_2 (u) |_R \, \phi\\
B_2^* (u) |_R \, \phi^*
\end{array}\right)
=\ti{\gamma}_1 (U) + \zeta_1 (U).\label{sdakns-b2-r}
\end{equation}
\end{lemma}

\begin{proof}
Direct calculation shows that \eqref{sdakns-b2-r} is true.
In the following we prove \eqref{sdakns-Am-r}.
Again, by direct verifications we can see that the relation
\begin{equation}\label{3.30}
A_m (u) |_R \, Q
= \mathcal{K}_{m,1} (U)
\end{equation}
is true for $m=1,2$.
We then assume \eqref{3.30} is correct for $A_m$ and $\mathcal{K}_{m,1}$.
For the $A_{m+1}$ case,
from \eqref{A_{m+1}-res} and using Lemma \ref{lem-3.2},
we have
\begin{equation*}
\begin{aligned}
A_{m+1}(u)|_R \, Q
=& \frac{1}{m}
\Big(
A_m^{\, \prime} (u)[\sigma_2 (u)] - B_2^{\, \prime} (u) [K_m (u)] + [A_m(u), B_2 (u)]
\Big) \Big|_R Q
-A_m (u)|_R \, Q \\
=& \frac{1}{m}
\Big(
 ( A_m |_R )'(U)[\ti{\gamma}_1] Q
-  ( B_2|_R )'(U)[\mathcal{K}_m] Q\\
&~~~~~~ + A_m(u) \big|_R B_2(u) \big|_R Q
- B_2(u) \big|_{R } A_m(u) \big|_R  Q
\Big)
-A_m (u)|_R \, Q.
\end{aligned}
\end{equation*}
Making use of \eqref{sdakns-b2-r}  and \eqref{3.30}, we have
\begin{equation*}
\begin{aligned}
A_{m+1}(u)|_R \, Q
=& \frac{1}{m}
\Big(
 ( A_m |_R)'(U)[\ti{\gamma}_1] Q
-  ( B_2|_R)'(U)[\mathcal{K}_m] Q\\
&~~~~~~ + A_m(u) \big|_R\, \ti{\gamma}_{1,1}+ A_m(u)|_R \,\zeta_{1,1}
- B_2(u) \big|_{R } \mathcal{K}_{m,1}
\Big)
-\mathcal{K}_{m,1}\\
=& \frac{1}{m}
\Big(
  (A_m|_R  Q  )'(U)[\ti{\gamma}_1]
- ( B_2\big|_R Q  )'(U)[\mathcal{K}_m]
+ A_m(u)|_R \,\zeta_{1,1}
\Big)
-\mathcal{K}_{m,1}.
\end{aligned}
\end{equation*}
Again, making use of \eqref{sdakns-b2-r} together with \eqref{3.30}, we  obtain
\begin{align*}
A_{m+1}(u)|_R \, Q
=& \frac{1}{m}
\Big(
\mathcal{K}_{m,1}^{\,\prime}(U)[\ti{\gamma}_1]
-\ti{\gamma}_{1,1}^{\,\prime} (U)[\mathcal{K}_m]
- \zeta_{1,1}^{\,\prime} (U)[\mathcal{K}_m]
+ A_m(u)|_R \,\zeta_{1,1}
\Big)
-\mathcal{K}_{m,1}\\
=& \mathcal{K}_{m+1,1}(U),
\end{align*}
where in the last step we have made use of \eqref{K-ga}  and \eqref{As}.

Now we have proved the first component of \eqref{sdakns-Am-r}.
Taking a similar procedure we can prove the second component.
Thus, we complete the proof of this lemma as well as the proof of Theorem \ref{Th-3-1}.

\end{proof}

\section{The sdBurgers hierarchies and Lie algebraic structure}\label{sec-4}

Similar to the continuous case \cite{CL-JPA-1992},
if we impose another constraint on $u$ (see equation \eqref{u=delta phi}),
we may obtain the sdBurgers hierarchy from \eqref{dkp-iso-lax-t}.
This fact can be proved by using the recursive structure \eqref{A_{m+1}-res}
together with the master symmetry and recursive structure of the sdBurgers hierarchy.
Thus, in this section, we consider the sdBurgers hierarchy
to formulate its recursive structure via Lie algebra and master symmetry.

\subsection{Isospectral and nonisopectral sdBurgers hierarchies}\label{sec-4-1}

The sdBurgers hierarchy can be obtained as a reduction of either the  relativistic Toda
hierarchy or the semi-discrete CLL hierarchy,
both of which are derived from the squared eigenfunction symmetry constraints
of the D$\Delta$mKP system \cite{CZZ-SAMP-2021,LZZ-SAMP-2024}.

To derive the sdBuegers hierarchy, we start from the following spectral problem \cite{CZZ-SAMP-2021,Z-PDEAM-2022},
\begin{subequations}\label{burgers-lax}
\begin{align}
&\Phi^{(n+1)} = M \Phi^{(n)},~~
\Phi^{(n)} =\left(\begin{array}{c}
\phi_{1}^{(n)}\\ \phi_{2}^{(n)}
\end{array}
\right),~~
M =\left(\begin{array}{cc}
\eta^2 + z^{(n)} & -\eta \\
\eta\, z^{(n+1)} & 0
\end{array}
\right),
\end{align}
together with the time part
\begin{align}
&\Phi^{(n)}_{t} = N \Phi^{(n)},~~
N =\left(\begin{array}{cc}
A  & B  \\
C  & D
\end{array}
\right),
\end{align}
\end{subequations}
where $\eta$ is the spectral parameter, $z \doteq z^{(n)}$ is a function of $(n,t)\in \mathbb{Z}\times \mathbb{R}$,
$A, B, C, D$ are undetermined functions of $n, t, z$ and $\eta$.
Note that in this section we use notation $z^{(n+j)}$ more than $E^j z$
for the sake of explicit expression.
The potential function $z$ satisfies
\begin{equation}\label{z-asym}
z \to 1,~~~ (n\to \pm \infty).
\end{equation}
The compatibility of the above system yields
\begin{equation}
M_t = (EN) M - M N,
\end{equation}
which gives rise to
\begin{subequations}\label{acd-b}
\begin{align}
&A^{(n)} = -\frac{1}{\eta} z^{(n)} B^{(n)}
+ 2 \Delta^{-1} \frac{\eta_t}{\eta}
+ a_0,\\
&C^{(n)} = -z^{(n)} B^{(n-1)},\\
&D^{(n)} = A^{(n+1)}
+ \eta B^{(n)}
+ \frac{1}{\eta} z^{(n)} B^{(n)}
- \frac{\eta_t}{\eta},
\end{align}
\end{subequations}
and
\begin{equation}\label{z_n,t}
z_{t} = \eta z E^{-1} \Delta B
-\frac{1}{\eta} z  \Delta z  B
+ 2 \frac{\eta_t}{\eta} z,
\end{equation}
where $a_0$ is an integration constant.
Inserting the expansion
\begin{equation}\label{expand-b}
B= \sum_{j=1}^{m} b_j \eta^{2(m-j)+1}
\end{equation}
into \eqref{z_n,t} and denoting operators
\begin{equation}\label{def-L}
L_1 = z  E^{-1} \Delta,~~
L_2 = z  \Delta z,~~
\bar{L} = L_2 {L_1}^{-1}
= z  \Delta z  \Delta^{-1} E \frac{1}{z},
\end{equation}
one can derive the isospectral sdBurgers equations by setting $\eta_{t_m} =0$:
\begin{align}
z_{t_m}=\bar{K}_{m}
= \bar{L}^{m-1} z  \Delta z ,~~m\ge 1,
\end{align}
where the first three read
\begin{subequations}
\begin{align}
&z_{t_1}=\bar{K}_1 =z \Delta z, \label{bur-k1}\\
&z_{t_2}=\bar{K}_2
=z^{(n)} \Delta (z^{(n+1)} z^{(n)}) - \bar{K}_1 ,\\
&z_{t_3}=\bar{K}_3
=z^{(n)}\Delta (z^{(n+2)} z^{(n+1)} z^{(n)})
-  \bar{K}_2- \bar{K}_1 .
\end{align}
\end{subequations}
The first equation \eqref{bur-k1} gives the sdBurgers equation \cite{CZZ-SAMP-2021,Z-PDEAM-2022}.
It is worth noting that when  calculating these flows, one should notice
the asymptotic condition \eqref{z-asym}. For example,  we have
\begin{equation}
\Delta^{-1} \Delta z
= \Delta^{-1} \Delta
(z  - 1)
=z  - 1.
\end{equation}
In light of the  asymptotic condition \eqref{z-asym} and the form of the recursion operator $\bar L$
defined in \eqref{def-L}, we know that
\begin{equation}
\bar{K}_m \to 0, ~~~ (n\to \pm \infty).
\end{equation}

To obtain nonisospectral equations of our interest,
we assume the spectral parameter $\eta$ evolves with respect to time as the following
\begin{equation}\label{eta-t}
\eta_{t_m} = \frac{1}{2}\sum^m_{j=0}(-1)^j \mathrm{C}^{j}_m \eta^{2(m-j)+1},~~~ (m\geq 1).
\end{equation}
Then, from \eqref{z_n,t} and \eqref{expand-b} we can derive the nonisospectral sdBurgers hierarchy
\begin{align}
z_{t_m}=\tilde{\sigma}_{m}
= (\bar{L}-1)^{m} z ,~~m \ge 1.
\end{align}
For convenience, we  introduce an auxiliary nonisospectral equation
\begin{equation}\label{sigma-0}
z_{t_0}=\tilde{\sigma}_{0}=z,
\end{equation}
which corresponds to
$\eta_{t_0} = \frac{1}{2} \eta$, $A=n, D=n+\frac{1}{2}, B=C=0$ in the Lax pair \eqref{burgers-lax}.

In the following we introduce
\begin{equation}\label{def-tiL}
\tilde{L} = \bar{L}-1
= z  \Delta z  \Delta^{-1} E \frac{1}{z } - 1,
\end{equation}
and consider  the combined sdBurgers hierarchies
\begin{subequations}\label{sdBurgers-flows}
\begin{align}
&z_{t_m}=\tilde{K}_{m }
= \tilde{L}^{m-1} \tilde{K}_1,~~~
\tilde{K}_1= z \Delta z,~~ ~~(m\ge 1), \label{sdBurgers-Km}\\
&z_{t_m}=\tilde{\sigma}_{m}
= \tilde{L}^m \tilde{\sigma}_0,~~~
\tilde{\sigma}_0=z,~~ ~~(m\ge 1).
\label{sdBurgers-Sigm}
\end{align}
\end{subequations}
Note that in light of \eqref{z-asym} these flows are asymptotically characterized by (for $m\ge 1$)
\begin{equation}
\tilde{K}_m \to 0, ~~ \tilde{\sigma}_m \to 0, ~~~ (n\to \pm \infty),
\end{equation}
which also makes possible to calculate higher order flows.\footnote{
From $\tilde{\sigma}_0=z$ it is easy to find $\bar{L}\tilde{\sigma}_0=z  \Delta (n z)$
and both of them tend to 1 when $n \to \pm\infty$.
This makes trouble in calculating the higher order flow
$\bar{L}^2\tilde{\sigma}_0=\bar{L}z  \Delta (n z)= z  \Delta z E \Delta^{-1}
\Delta (n  z)$.
We do not formally write $\Delta^{-1}\Delta (nz)=nz$ without considering any asymptotic condition.
}
The first few of them are
\begin{subequations}
\begin{align}
&\tilde{K}_1
=z\Delta z,\label{sdburgers-k1}\\
&\tilde{K}_2
=z^{(n)}\Delta (z^{(n+1)} z^{(n)})-2\tilde{K}_1,\\
&\tilde{K}_3
=z^{(n)}\Delta (z^{(n+2)}z^{(n+1)} z^{(n)})
-3\tilde{K}_2-3\tilde{K}_1,\\
&\tilde{\sigma}_1
= n\ti{K}_1 +z^{(n+1)} z^{(n)} -z^{(n)},\label{sdBurgers-sigma_1}\\
&\ti{\sigma}_2
=(n+1) \tilde{K}_2 +\tilde{K}_1 +z^{(n+2)}z^{(n+1)} z^{(n)} - 2z^{(n+1)} z^{(n)} + z^{(n)}.
\end{align}
\end{subequations}

Later, we will see that the isospectral hierarchy \eqref{sdBurgers-Km} can be recovered from
certain constraint of the D$\Delta$KP system.

\subsection{Algebraic structure  and  master symmetry}\label{sec-4-2}

To achieve algebraic structure of $\{\tilde{K}_m, \tilde{\sigma}_m\}$
and identify the master symmetry, we make use of some properties of the recursion operator $\tilde L$.
Let us introduce
\begin{equation}\label{w}
w=z-1
\end{equation}
and it follows that $w$ goes to zero rapidly as $|n| \to \infty$.
By $S[w]$ we denote the Schwartz space composed by all functions $f(w)$
that are $C^{\infty}$ differentiable with respect to $w$ and its shifts.

For the sdBurgers equation
\begin{equation}\label{sdbur-1}
z_{t_1}=\tilde K_1,
\end{equation}
one can verify
\begin{equation}\label{L'K-1}
\tilde{L}'(z)[\ti{K}_1] = [\ti{K}_1', \tilde{L}].
\end{equation}
Note that here and after in this section the G$\hat{\mathrm{a}}$teaux derivatives are taken with respect to $z$,
and we simply write $f'[g]$ instead of $f'(z)[g]$ without making confusion.
The above relation \eqref{L'K-1} means $\tilde{L}$ is a strong symmetry of equation \eqref{sdbur-1}.
In addition, via long but straight forward verification we get
\begin{equation}
	\tilde{L}'[\tilde{L} f]g-\tilde{L}'[\tilde{L} g] f =\tilde{L} (\tilde{L}'[f]g-\tilde{L}'[g]f),~~\forall f,g \in S[w],
\end{equation}
i.e., $\tilde{L}$ is a hereditary operator.
This means $\tilde{L}$ is a strong symmetry for any equation in the isospectral sdBurgers
hierarchy \eqref{sdBurgers-Km}, i.e.
\begin{equation}\label{4.23}
\tilde{L}'[\ti{K}_m] = [\ti{K}_m', \tilde{L}],~m=1,2,\cdots.
\end{equation}
It then follows that the isospectral flow $\tilde K_s$ is a symmetry of the hierarchy \eqref{sdBurgers-Km}
and we have
\begin{equation}\label{Km-Kn}
\llbracket \ti{K}_m, \ti{K}_s \rrbracket = 0,~~m,s=1,2,\cdots.
\end{equation}

Next, we look for the algebraic structures related to nonisospectral flows.
Let us start from the following lemma

\begin{lemma}\label{lem-4.1}
With $\ti{L}$ defined in \eqref{def-tiL} and the nonisospectral flows $\{\tilde{\sigma}_m\}$
defined in \eqref{sdBurgers-Sigm},
we have
\begin{equation}\label{L'sig_n}
\tilde{L}'[\tilde{\sigma}_m]
= [\tilde{\sigma}_m', \tilde{L}] + \tilde{L}^{m+1} + \tilde{L}^m, ~~m=0,1,2,\cdots.
\end{equation}
\end{lemma}

\begin{proof}
Since $\tilde{\sigma}_0=z$ we have $\tilde{\sigma}_0'=I$ where $I$ stands for an identity operator,
which yields
$[\tilde{\sigma}_0', \tilde{L}]
=[I,\tilde{L}]
= 0$.
Meanwhile, direct calculation shows
\begin{equation}
\tilde{L}'[\tilde{\sigma}_0]
= \tilde{L} + I,\label{L'sigma-0}
\end{equation}
which means \eqref{L'sig_n} is true for $m=0$.
In addition, direct calculation also shows \eqref{L'sig_n} is true for $m=1$.
We prove the general case by induction.
Assuming
\begin{equation}
\tilde{L}'[\tilde{\sigma}_{m-1}]
= [\tilde{\sigma}_{m-1}', \tilde{L}] + \tilde{L}^{m} + \tilde{L}^{m-1},
\end{equation}
then for arbitrary function $g \in S[w]$, we have
\begin{align*}
\big(
\tilde{L}'[\tilde{\sigma}_m] - [\tilde{\sigma}_m' , \tilde{L}]
\big)g
=& \tilde{L}'[\tilde{\sigma}_m] g
- \tilde{\sigma}_m' \tilde{L} g +\tilde{L} \tilde{\sigma}_m' g\\
=& \ti{L}'[\ti{L}\ti{\sigma}_{m-1}] g
-(\ti{L}\ti{\sigma}_{m-1})'\ti{L}g
+ \ti{L} (\ti{L}\ti{\sigma}_{m-1})' g\\
=& \ti{L}'[\ti{L}\ti{\sigma}_{m-1}] g
- \ti{L}'[\ti{L}g]\ti{\sigma}_{m-1}
- \ti{L}\ti{\sigma}_{m-1}'[\ti{L}g]
 +\ti{L}\ti{L}'[g]\ti{\sigma}_{m-1}
+\ti{L}\ti{L}\ti{\sigma}_{m-1}' [g].
\end{align*}
Noticing that $\ti{L}$ is heredity, the above relation reduces to
\begin{align*}
\big(
\tilde{L}'[\tilde{\sigma}_m] - [\tilde{\sigma}_m' , \tilde{L}]
\big)g
= \ti{L} \big(
\ti{L}'[\ti{\sigma}_{m-1}]-[\ti{\sigma}_{m-1}', \ti{L}]
\big) g
=  \ti{L} \big( \ti{L}^m + \ti{L}^{m-1}\big) g
= \big( \ti{L}^{m+1} + \ti{L}^{m}\big) g,
\end{align*}
which proves the general case.

\end{proof}

Using this lemma we can prove the following.

\begin{lemma}\label{lem-4.2}
For the flows $\{ \tilde{K}_{m }, \tilde{\sigma}_{s}\}$ define by \eqref{sdBurgers-flows},
we have the following relation:
\begin{equation}\label{Km-sig_n}
\llbracket \ti{K}_m, \ti{\sigma}_s \rrbracket = m (\ti{K}_{m+s} + \ti{K}_{m+s-1}),
\end{equation}
where  $m \ge 1$, $s \ge 0$ and  we set $\ti{K}_0 = 0$.
\end{lemma}

\begin{proof}
Let us first prove
\begin{equation}\label{K1-sig_n}
\llbracket \ti{K}_1, \ti{\sigma}_s \rrbracket = \ti{K}_{s+1} + \ti{K}_s, ~~s= 0, 1, 2, \cdots.
\end{equation}
Direct calculation yields
\begin{equation}
\llbracket \ti{K}_1, \ti{\sigma}_0 \rrbracket =  \ti{K}_1,~~~
\llbracket \ti{K}_1, \ti{\sigma}_1 \rrbracket = \ti{K}_2 + \ti{K}_1.
\end{equation}
Assuming
\begin{equation}\label{4.31}
\llbracket \ti{K}_1, \ti{\sigma}_{s-1} \rrbracket = \ti{K}_{s} + \ti{K}_{s-1},
\end{equation}
in which we set $\ti{K}_{0}=0$,
then we have
\begin{align*}
\llbracket \ti{K}_1, \ti{\sigma}_s \rrbracket
&= \llbracket \ti{K}_1, \ti{L} \ti{\sigma}_{s-1} \rrbracket\\
&= \ti{K}_1'  \ti{L} \ti{\sigma}_{s-1}
- (\ti{L} \ti{\sigma}_{s-1})'[ \ti{K}_1]\\
&= \ti{K}_1' \ti{L} \ti{\sigma}_{s-1}
- \ti{L}' \ti{K}_1 \ti{\sigma}_{s-1}
- \ti{L} \ti{\sigma}_{s-1}' [\ti{K}_1].
\end{align*}
Using relation \eqref{4.23} we further have
\begin{align*}
\llbracket \ti{K}_1, \ti{\sigma}_s \rrbracket
&= \ti{K}_1' \ti{L} \ti{\sigma}_{s-1}
- [\ti{K}_1', \ti{L}] \ti{\sigma}_{s-1}
- \ti{L} \ti{\sigma}_{s-1}' [\ti{K}_1]\\
&= \ti{L} \ti{K}_1' [\ti{\sigma}_{s-1}]
- \ti{L} \ti{\sigma}_{s-1}' [\ti{K}_1]\\
&= \ti{L} \llbracket \ti{K}_1, \ti{\sigma}_{s-1} \rrbracket\\
&= \ti{L} (\ti{K}_s + \ti{K}_{s-1})\\
&= \ti{K}_{s+1} + \ti{K}_s,
\end{align*}
where \eqref{4.31} is used. Thus, \eqref{K1-sig_n} is proved.

Relation \eqref{K1-sig_n} indicates \eqref{Km-sig_n} holds for $m=1$.
Next, we prove the general case with respect to $m$ by induction.
We assume
\begin{equation}\label{4.32}
\llbracket \ti{K}_{m-1}, \ti{\sigma}_n \rrbracket
= (m-1)(\ti{K}_{m+n-1} + \ti{K}_{m+n-2}).
\end{equation}
Then, we have
\begin{align*}
\llbracket \ti{K}_m, \ti{\sigma}_s \rrbracket
&= \llbracket \ti{L} \ti{K}_{m-1}, \ti{\sigma}_s \rrbracket\\
&= \big( \ti{L} \ti{K}_{m-1} \big)' \ti{\sigma}_s
- \ti{\sigma}_s' \big( \ti{L} \ti{K}_{m-1} \big) \\
&= \ti{L}' [\ti{\sigma}_s] \ti{K}_{m-1}
+ \ti{L} \ti{K}_{m-1}'[\sigma_s]
- \ti{\sigma}_s' \ti{L} \ti{K}_{m-1}.
\end{align*}
Inserting  \eqref{L'sig_n} gives rise to
\begin{align*}
\llbracket \ti{K}_m, \ti{\sigma}_s \rrbracket
&= \big(
[\ti{\sigma}_s', \ti{L}]
+ \ti{L}^{s+1}
+\ti{L}^s \big) \ti{K}_{m-1}
+ \ti{L} \ti{K}_{m-1}'[\ti{\sigma}_s]
- \ti{\sigma}_s' \ti{L} \ti{K}_{m-1}\\
&= [\ti{\sigma}_s', \ti{L}]  \ti{K}_{m-1}
+ \ti{K}_{m+s} + \ti{K}_{m+s-1}
+ \ti{L} \ti{K}_{m-1}'[\ti{\sigma}_s]
- \ti{\sigma}_s' \ti{L} \ti{K}_{m-1}\\
&= \ti{K}_{m+s} + \ti{K}_{m+s-1}
+ \ti{L} \llbracket \ti{K}_{m-1}, \ti{\sigma}_s \rrbracket\\
&= \ti{K}_{m+s} + \ti{K}_{m+s-1}
+ \ti{L} (m-1) (\ti{K}_{m+s-1} + \ti{K}_{m+s-2})\\
&= m (\ti{K}_{m+s} + \ti{K}_{m+s-1}),
\end{align*}
where \eqref{4.32} is used.
Thus the proof of the lemma is completed.

\end{proof}

We also have the follow structure for  the nonisospectral flows.

\begin{lemma}\label{lem-4.3}
The nonisospectral flows defined in \eqref{sdBurgers-Sigm} satisfy
\begin{equation}\label{sig-m,sig-n}
\llbracket \ti{\sigma}_m, \ti{\sigma}_s \rrbracket = (m-s)(\ti{\sigma}_{m+s} + \ti{\sigma}_{m+s-1}),
~~m\geq 1,~s\geq 0.
\end{equation}
\end{lemma}

The proof is similar to the one for Lemma \ref{lem-4.2}. Here we skip it.

Now we come to the stage to summarize the algebraic structure for the flows
$\{\ti{K}_m, \ti{\sigma}_s\}$.

\begin{theorem}\label{Th-4-1}
The flows $\{\ti{K}_m, \ti{\sigma}_s\}$ defined by \eqref{sdBurgers-flows} span a Lie algebra with  structure
\begin{subequations}\label{lie-sdburgers}
\begin{align}
&\llbracket \ti{K}_m, \ti{K}_s \rrbracket = 0,\\
&\llbracket \ti{K}_m, \ti{\sigma}_s \rrbracket = m(\ti{K}_{m+s} + \ti{K}_{m+s-1}),\label{sdburgers-Km-tau-r}\\
&\llbracket \ti{\sigma}_m, \ti{\sigma}_s \rrbracket = (m-s) (\ti{\sigma}_{m+s} + \ti{\sigma}_{m+s-1}),
\label{4.34c}
\end{align}
\end{subequations}
where $m \geq 1$, $s \geq 0$ and $\ti{K}_0(z)$ is set to be $0$.
The above structure indicates $\ti{\sigma}_1$ play a role of master symmetry for the
sdBurgers flows by
\begin{subequations}
\begin{align}
&\ti{K}_{s+1}
= \frac{1}{s} \llbracket \ti{K}_s, \ti{\sigma}_1 \rrbracket
- \ti{K}_{s}, ~~s\geq 1, \label{4.35a}\\
&\ti{\sigma}_{s+1} = \frac{1}{s-1} \llbracket \ti{\sigma}_s, \ti{\sigma}_1 \rrbracket - \ti{\sigma}_s, ~~s> 1.	
\end{align}
\end{subequations}

\end{theorem}

\begin{corollary}\label{cor-4-1}
In the isospectral sdBurgers hierarchy \eqref{sdBurgers-Km}, each equation $z_{t_s} = \ti{K}_s$
possesses
two sets of symmetries
\begin{equation}
\{\ti{K}_l\},~~
\{\ti{\tau}_{s,r}
=s t_s (\ti{K}_{s+r}
+\ti{K}_{s+r-1} )
+\ti{\sigma}_r \},
\end{equation}
which generate a Lie algebra with   structure
\begin{subequations}\label{lie-sdburgers-2}
\begin{align}
&\llbracket \ti{K}_l, \ti{K}_r \rrbracket = 0,\\
&\llbracket \ti{K}_l, \ti{\tau}_{s,r} \rrbracket = l(\ti{K}_{l+r} + \ti{K}_{l+r-1}),\\
&\llbracket \ti{\tau}_{s,l}, \ti{\tau}_{s,r} \rrbracket = (l-r) (\ti{\tau}_{s,l+r} + \ti{\tau}_{s,l+r-1}),
\end{align}
\end{subequations}
where $l,s \geq 1$, $ r \geq 0$, and $\ti{K}_0 $ is set to be $0$.
\end{corollary}

\section{Constraint of D$\Delta$KP system to sdBurgers hierarchy}\label{sec-5}

In  continuous case, it is known that the time evolution of the
eigenfunction for the KP hierarchy can be converted to the Burgers hierarchy
under certain constraint \cite{CL-JPA-1992}.
We have similar results in the semi-discrete case.

In Sec.\ref{sec-2-2} we already presented the D$\Delta$KP hierarchy \eqref{dkp-Km} together with
the   evolutions \eqref{dkp-phi_x} and \eqref{dkp-iso-lax-t} for the eigenfunction $\phi$.
We introduce a linear eigenfunction constraint between $u$ and $\phi$:
\begin{equation}\label{u=delta phi}
u = \Delta \phi,
\end{equation}
where we assume $\phi \to 0,~ |n|\to \infty$. Define $z$ by
\begin{equation}\label{def-w_n}
 \phi=z- 1,
\end{equation}
under which the $x$-evolution \eqref{dkp-phi_x} gives rise to
\begin{equation}\label{phi_x-R}
z_x = z \Delta z,
\end{equation}
which is the sdBurgers equation \eqref{bur-k1} with $x=t_1$.
Using the above three equations one can replace all the derivatives $\partial_x^{j}u^{(n+i)}$
in the D$\Delta$KP flows $\{K_m\}$ and operators $\{A_m\}$,
and then express them in terms of only $z$ and its shifts.
For convenience, we denote
\begin{equation}
F(u) |_{R_1}= F(u) |_{\eqref{u=delta phi},\eqref{def-w_n},\eqref{phi_x-R}}.
\end{equation}
Then we have the following results.

\begin{theorem}\label{Th-5-1}
Under the constraints \eqref{u=delta phi}, \eqref{def-w_n} and \eqref{phi_x-R},
both the isospectral  D$\Delta$KP hierarchy \eqref{dkp-Km} and the time  evolution \eqref{dkp-iso-lax-t}
give rise to the sdBurgers hierarchy \eqref{sdBurgers-Km} respectively.
\end{theorem}

We prove this theorem via the following lemmas.

\begin{lemma}\label{lem-5-1}
For the flows $\{K_m\}$, $\sigma_2$,   $\{\ti{K}_m\}$ and $\ti{\sigma}_1$
defined in \eqref{dkp-Km}, \eqref{dkp-sigma-2},  \eqref{sdBurgers-Km} and \eqref{sdBurgers-sigma_1}, respectively,
under the constraints \eqref{u=delta phi}, \eqref{def-w_n} and \eqref{phi_x-R},
we have
\begin{subequations}
\begin{align}
& K_m(u)|_{R_1} =\Delta \ti{K}_m(z), \label{5.5a}\\
& \sigma_2(u)|_{R_1} =\Delta \ti{\ti{\sigma}}_1(z), \label{5.5b}
\end{align}
\end{subequations}
where
\begin{equation}
\ti{\ti{\sigma}}_1 = x \ti{K}_2 + x \ti{K}_1 + \ti{\sigma}_1 - \ti{K}_1.
\end{equation}
\end{lemma}

\begin{proof}
For the first flow $K_1(u)$, under the constraints \eqref{u=delta phi} \eqref{def-w_n} and \eqref{phi_x-R},
it is easy to have
\begin{equation}
K_1(u) |_{R_1} = u_x |_{R_1}
= \Delta \phi_{x}
= \Delta (z \Delta z)
= \Delta \ti{K}_1(z).
\end{equation}
For the second flow \eqref{dkp-k2}, by direct calculation, it is not difficult to  see that
\eqref{5.5a} holds for $m=2$ as well.
With these in hand, one can check the validity of  \eqref{5.5b},
while here we skip the details.

In the next we look at \eqref{5.5a} for general $m$.
Assume we have
\begin{equation}\label{5.9}
K_m(u)|_{R_1} =\Delta  \ti{K}_m(z),
\end{equation}
and we are going to prove it is true for $K_{m+1}(u)$ and $\ti{K}_{m+1}(z)$.
It  indicates
\begin{equation}\label{5.10}
(K_m(u)|_{R_1})'(z) [f]=\Delta  (\ti{K}_m^{\,\prime}(z)[f]), ~~~ f\in S[w],
\end{equation}
where $w$ is introduced in \eqref{w} and we note that here the G$\hat{\mathrm{a}}$teaux derivatives
on both sides are all taken with respect to $z$.
Note also that, in light of \eqref{sdburgers-Km-tau-r}, we have
\begin{equation}\label{Km-titi-sigma-1}
\llbracket \ti{K}_m, \ti{\ti{\sigma}}_1 \rrbracket
= \llbracket \ti{K}_m, \ti{\sigma}_1 \rrbracket
= m (\ti{K}_{m+1} + \ti{K}_m).
\end{equation}
Now we calculate
\begin{align*}
K_m^{\,\prime} (u)[\sigma_2 (u)] |_{R_1}
&=\Bigl(\frac{\mathrm{d}}{\mathrm{d}\,\varepsilon}
K_m( u +\varepsilon \sigma_2)\big|_{\varepsilon=0} \Bigr)\Bigr|_{R_1}
\\
&=\frac{\mathrm{d}}{\mathrm{d}\,\varepsilon}
{K_m( \Delta \phi +\varepsilon \Delta \ti{\ti{\sigma}}_1) }\big|_{\varepsilon=0}\\
&=\frac{\mathrm{d}}{\mathrm{d}\,\varepsilon}
{K_m( \Delta (z +\varepsilon \ti{\ti{\sigma}}_1)) }\big|_{\varepsilon=0}\\
&=(K_m(u)|_{R_1})'(z) [\ti{\ti{\sigma}}_1 (z)]\\
&= \Delta \bigl(\ti{K}_m^{\,\prime} (z) [\ti{\ti{\sigma}}_1 (z)]\bigr),
\end{align*}
where we have made use of \eqref{5.5b} and \eqref{5.9}.
On the other hand, in a similar way, we get
\begin{equation*}
\sigma_2^{\prime} (u)  [K_m (u)] |_{R_1}
= \Delta \big(\ti{\ti{\sigma}}_1{}^{\prime} (z)  [\ti{K}_m (z)]\big).
\end{equation*}
Thus, from  \eqref{dkp-master-sym} and the above results, we have
\begin{align*}
K_{m+1} (u)|_{R_1}
&= \frac{1}{m}
\llbracket K_m(u), \sigma_2(u)\rrbracket \big|_{R_1}
- K_m(u) |_{R_1}\\
&= \frac{1}{m}\Big(K_m^{\,\prime}(u)[\sigma_2(u)]- \sigma_2^{\,\prime}(u)[K_m(u)]\Bigr )
- K_m(u) |_{R_1}\\
&= \frac{1}{m}\Big(
\Delta \bigl(\ti{K}_m^{\,\prime} (z)[\ti{\ti{\sigma}}_1 (z)]\bigr)
- \Delta \bigl(\ti{\ti{\sigma}}_1{}^{\prime} (z) [\ti{K}_m (z)]\bigr)
\Big)
- \Delta \ti{K}_m (z)\\
&= \Delta \bigg(
\frac{1}{m}
\llbracket \ti{K}_m (z),
\ti{\ti{\sigma}}_1 (z) \rrbracket
- \ti{K}_m (z)
\bigg)\\
&= \Delta \ti{K}_{m+1} (z),
\end{align*}
where we have made use of \eqref{5.9}, \eqref{Km-titi-sigma-1} and \eqref{4.35a}.
Thus, \eqref{5.5a} is true for general $m$.

\end{proof}

\begin{lemma}\label{lem-5-2}
Under the constraints \eqref{u=delta phi}, \eqref{def-w_n} and \eqref{phi_x-R}, we have
\begin{subequations}
\begin{align}
&A_m (u) |_{R_1} \, \phi
= \ti{K}_m (z), \label{5.11a}\\
&B_2 (u) |_{R_1} \, \phi
= \ti{\ti{\sigma}}_1 (z) - \phi
= \ti{\ti{\sigma}}_1 (z) - z + 1, \label{5.11b}
\end{align}
\end{subequations}
where $B_2$ is given in \eqref{B_2}.
\end{lemma}

\begin{proof}
Direct calculation shows that \eqref{5.11a} holds for $m=1,2$.
For \eqref{5.11b}, there is
\begin{equation*}
\begin{aligned}
B_2 (u) |_{R_1} \, \phi
&= x \ti{K}_2 (z)
+ (x+n)\ti{K}_1 (z)
+(z - 1 )^2 \\
&= x \ti{K}_2 (z) + x \ti{K}_2 (z) + \ti{\sigma}_1 - \ti{K}_1 - z + 1 \\
&= \ti{\ti{\sigma}}_1 - z + 1,	
\end{aligned}
\end{equation*}
which means \eqref{5.11b} holds.

Then we assume
\begin{equation}\label{5.12}
A_m(u)|_{R_1} \, \phi
=\ti{K}_m (z)
\end{equation}
and in the next we go to prove it to be true for $A_{m+1}$ and $\ti{K}_{m+1}$ as well.
It is easy to get
\begin{align}
A_m'[\sigma_2] |_{R_1}
&= \Bigl(\frac{\mathrm{d}}{\mathrm{d}\,\varepsilon}
A_m( u +\varepsilon \sigma_2)\big|_{\varepsilon=0}\Bigr)\Big|_{R_1}
\nonumber \\
&= \frac{\mathrm{d}}{\mathrm{d}\,\varepsilon}
{A_m( \Delta \phi +\varepsilon \Delta \ti{\ti{\sigma}}_1)} \big|_{\varepsilon=0} \nonumber \\
&= \frac{\mathrm{d}}{\mathrm{d}\,\varepsilon}
{A_m( \Delta (z +\varepsilon \ti{\ti{\sigma}}_1))}\big|_{\varepsilon=0} \nonumber \\
&= \big( A_m |_{R_1}\big)'(u)[\ti{\ti{\sigma}}_1 (z)], \label{Am-r}
\end{align}
and similarly
\begin{equation}\label{B2-r}
B_2^{\,\prime} (u)[K_m (u)] |_{R_1}
= \big( B_2 |_{R_1}\big)'(z)[\ti{K}_m (z)].
\end{equation}
Now, using the relation \eqref{A_{m+1}-res}  we calculate $A_{m+1}(u)|_{R_1}  \phi$:
\begin{align*}
A_{m+1}(u)|_{R_1} \phi
=& \frac{1}{m}
\big(
A_m ^{\prime}(u)[\sigma_2] - B_2^{\prime}(u)[K_m ] + [A_m, B_2]
\big) \big|_{R_1} \phi
-A_m|_{R_1}  \phi\\
=& \frac{1}{m}
\Big(
\big(A_m |_{R_1}\big)'(z)[\ti{\ti{\sigma}}_1 (z)] \phi
- \big( B_2|_{R_1}\big)'(z) [\ti{K}_m (z)] \phi \\
&~~~~~ + A_m |_{R_1} B_2 |_{R_1} \phi - B_2 |_{R_1} A_m |_{R_1}  \phi
\Big)
-A_m |_{R_1}  \phi,
\end{align*}
where we have substituted \eqref{Am-r} and \eqref{B2-r}.
By noticing that $\phi'(z)[f]=z'(z)[f]=f$ for $f\in S[w]$ and also  inserting \eqref{5.12}, it is rewritten as
\begin{align*}
A_{m+1}(u)|_{R_1} \, \phi
=& \frac{1}{m}
\Big( \big( A_m |_{R_1} \,\phi  \big)'(z)[\ti{\ti{\sigma}}_1 (z)]
-  \big( B_2 |_{R_1} \phi \big)'(z)[\ti{K}_m (z)]
- A_m |_{R_1} \phi
\Big)
-A_m |_{R_1}  \phi\\
=& \frac{1}{m}
\Big( \ti{K}_m^{\,\prime} (z) [\ti{\ti{\sigma}}_1 (z)]
-  \big( B_2 |_{R_1} \phi \big)'(z)[\ti{K}_m (z)]
- \ti{K}_m (z)
\Big)
-\ti{K}_m (z).
\end{align*}
Next, substituting  \eqref{5.11b}, we arrive at
\begin{align*}
A_{m+1}(u)|_{R_1} \, \phi
=& \frac{1}{m}
\Big( \ti{K}_m^{\,\prime} (z) [\ti{\ti{\sigma}}_1 (z)]
- \ti{\ti{\sigma}}_1^{\,\prime} (z) [\ti{K}_m (z)]
+ z'(z)[\ti{K}_m (z)]
-\ti{K}_m (z)
\Big)
-\ti{K}_m (z)\\
=& \frac{1}{m}
\llbracket \ti{K}_m (z), \ti{\ti{\sigma}}_1 (z) \rrbracket
-\ti{K}_m (z)\\
=& \ti{K}_{m+1} (z),
\end{align*}
where in the last step the relations  {\eqref{Km-titi-sigma-1}} and \eqref{4.35a} are used.
Thus, we have proved \eqref{5.11a} holds for general $m$.

\end{proof}

Based on these two lemmas, Theorem \ref{Th-5-1} is proved.

\section{Constraint of D$\Delta$mKP system to sdBurgers hierarchy}\label{sec-6}

In the continuous case, the squared eigenfunction constraint of the mKP system
gives rise to the CLL hierarchy which can be further reduced to the Burgers hierarchy \cite{CL-JPA-1992,C-JMP-2002}.
There are more results in semi-discrete case:
through two different squared eigenfunction symmetry constraints, the D$\Delta$mKP system
can yield the relativistic Toda hierarchy (see \cite{CZZ-SAMP-2021})
and the semi-discrete CLL hierarchy (see \cite{LZZ-SAMP-2024}),
both of which can be reduced to the sdBurgers hierarchy \cite{CZZ-SAMP-2021,LZZ-SAMP-2024}).
In the following we show a linear eigenfunction constraint
that leads the D$\Delta$mKP hierarchy to the sdBurgers hierarchy.

\subsection{D$\Delta$mKP system to sdBurgers hierarchy and algebraic structure}\label{sec-6-1}

We first quickly review the D$\Delta$mKP hierarchy.
One can refer to \cite{CZZ-SAMP-2021} for more details.
Consider the pseudo-difference operator
\begin{equation}\label{L-dmKP}
M = v\Delta + v_0 + v_1\Delta^{-1} + \cdots + v_j \Delta^{-j}+\cdots,
\end{equation}
and the isospectral Lax triad \cite{CZZ-SAMP-2021}
\begin{subequations}\label{dmkp-iso-lax}
\begin{align}
& M \psi = \eta \psi,~~~\eta_{t_m}=0, \label{dmkp-sp}\\
& \psi_x = \hat{A}_1 \psi,~~~\hat{A}_1=v\Delta, \label{dmkp-iso-A1}\\
& \psi_{t_m}=\hat{A}_m \psi,~~~m=1, 2, \cdots, \label{dmkp-flow-lax}
\end{align}
\end{subequations}
where $v$ and $\{v_j\}$ are functions of $(n, x, \mathbf{t}=(t_1,t_2,\cdots,))$ satisfying
\begin{equation}\label{6.3}
v \to 1,~~ v_j \to 0, ~~~ |n|\to \infty,
\end{equation}
$\ti{A}_s=(M^s)_{\geq 1}$ contains the pure difference part (in terms of $\Delta$)
of $M^s$,
satisfying the asymptotic condition 
$\ti{A}_s \to \Delta^s, ~ |n|\to \infty$,
and the first three are
\begin{subequations}\label{dmkp-Aj}
\begin{align}
&\hat{A}_1=v\Delta,\label{dmkp-A1}\\
&\hat{A}_2=v(E v)\Delta^2+v(\Delta v)\Delta+v(E v_0)\Delta+v v_0\Delta,\\
&\hat{A}_3
= \hat{a}_1\Delta^3+\hat{a}_2\Delta^2+\hat{a}_3\Delta,
\end{align}
\end{subequations}
in which
\begin{align*}
\hat{a}_1&=v (E v)(E^2 v),\\
\hat{a}_2&=2 v(E v)(E\Delta v)+ v(E v)(E^2 v_0)+ v(\Delta  v)(E v)+
v({E} v_0)(E v)
+ v v_0(E v),\\
\hat{a}_3&= v(E v)(\Delta^2 v)+2 v(E v)(E\Delta v_0)+v(E v)(E^2 v_1)+ v(\Delta v)^2+2v(\Delta v)(E v_0)\\
&~~~+ v(E v_0)^2+ v v_0(\Delta v)
+  {v} v_0(E v_0)
+ v^2(\Delta v_0)
+{v}^2(E v_1)+v v_1(E^{-1} v)
+{v_0^2 v}.
\end{align*}
The compatibility  of \eqref{dmkp-iso-lax} gives rise to
\begin{subequations}\label{dmkp-iso-com}
\begin{align}
&M_x=[\hat{A}_1,M],\label{dmkp-iso-Lx}\\
&M_{t_m}=[\hat{A}_m,M],\label{dmkp-iso-Ltm}\\
&\hat{A}_{1,t_m}-\hat{A}_{m,x}+[\hat{A}_1,\hat{A}_m]=0,\label{dmkp-iso-zcc}
\end{align}
\end{subequations}
Among them, equation \eqref{dmkp-iso-Lx} yields the expression of $v_j$ in terms of $v$,
e.g.
\begin{subequations}\label{dmkp-v}
\begin{align}
&v\Delta v_0=v_x,\label{dmkp-v0}\\	
&v(Ev_1)-(E^{-1}v)v_1
=v_{0,x}-v(\Delta v_0), \label{dmkp-v1}
\end{align}
\end{subequations}
from which one can `integrate' $v_j, ~j \geq 1$ with boundary condition \eqref{6.3},
and then one has
\begin{align}\label{dmkp-v0-v1}
v_0=\Delta^{-1}(\ln v)_x,~~~
v_1=\frac{\Delta^{-2}(\ln v)_{xx}
-\Delta^{-1}v_x}{(E^{-1}v)}.
\end{align}
Eq.\eqref{dmkp-iso-zcc} provides the zero curvature representation of the isospectral D$\Delta$mKP hierarchy,
\begin{equation}\label{dmkp-iso-km}
v_{t_m}=\hat{K}_m(v)
= (\hat{A}_{m,x}-[\hat{A}_1, \hat{A}_m])\Delta^{-1},~~~m=1,2,\cdots,
\end{equation}
where the first three of them are
\begin{subequations}\label{dmkp-k}
\begin{align}
v_{t_1}=\hat{K}_1(v)&=v_x, \label{6.9a}\\
v_{t_2}=\hat{K}_2(v)
&=v(1+2\Delta^{-1})(\ln v)_{xx}
+v_{x}(1+2\Delta^{-1})(\ln v)_{x}-2 v v_{x}, \label{6.9b}\\
v_{t_3}=\hat{K}_3(v)
&=v (3\Delta^{-2}+3\Delta^{-1}+1)(\ln v)_{xxx}+ v_{x}(3\Delta^{-2}+3\Delta^{-1}+1)
(\ln v)_{xx}\nonumber\\
&~~~+v_{x}((1+2\Delta^{-1})(\ln v)_{x})^2+2v((1+2\Delta^{-1})(\ln v)_{x})((1+2\Delta^{-1})
(\ln v)_{xx})\nonumber\\
&~~~-v_{x}(\Delta^{-1}(\ln v)_{x})((\Delta^{-1}+1)(\ln v)_{x})
-v(\Delta^{-1}(\ln v)_{xx })((\Delta^{-1}+1)(\ln v)_{x})\nonumber\\
&~~~-v(\Delta^{-1}(\ln v)_{x})((\Delta^{-1}+1)(\ln v)_{xx})+3 v^2v_{x}-3v(1+\Delta^{-1})v_{xx}\nonumber\\
&~~~-6v v_{x}\Delta^{-1}(\ln v)_{x}-3 v^2\Delta^{-1}(\ln v)_{xx}-3v_{x}(1+\Delta^{-1})v_{x}
\end{align}
\end{subequations}

To get the nonisospectral flows $\{\hat{\sigma}_s\}$,
we consider the nonisospectral Lax triad \cite{LZZ-SAMP-2024}:
\begin{subequations}\label{dmkp-non-lax}
\begin{align}
&M \hat \psi=\eta  \psi,~~~\eta_{t_s}= \eta^s+\eta^{s-1},\label{4.11}\\
&\psi_x=\hat{A}_1 \psi,\\
&\psi_{t_s}=\hat{B}_s \psi,~~~s=2, 3, \cdots,
\end{align}
\end{subequations}
where $\hat B_s$ is assumed to be a difference operator with the form
\begin{equation*}
\hat{B}_s=\sum_{j=0}^{s-1}\hat{b}_j\Delta^{s-j},
\end{equation*}
with the asymoptotic condition
\begin{align}\label{dmkp-non-bc}
{\hat{B}_s}|_{v=1}=
x\Delta^s+(x+n)\Delta^{s-1}, ~~~ s \ge 2.
\end{align}
The compatibility of \eqref{dmkp-non-lax} gives rise to
\begin{subequations}\label{dmkp-non-com}
\begin{align}
&M_{x}=[\hat{A}_1,M],\label{dmkp-non-Lx}\\
&M_{t_s}=[\hat{B}_s,M]+M^s+M^{s-1},\label{dmkp-non-Ltm}\\
&\hat{A}_{1,t_s}-\hat{B}_{s,x}+[\hat{A}_1,\hat{B}_s]=0.\label{dmkp-non-zcc}
\end{align}
\end{subequations}
$\{\hat B_m\}$ are uniquely determined from \eqref{dmkp-non-Ltm} together with the boundary condition
\eqref{dmkp-non-bc},
but no $\hat{B}_1$ can be obtained in this approach, therefore there is no the
nonisospectral D$\Delta$mKP flow $\hat\sigma_1$ is defined from \eqref{dmkp-non-com}.
The first two $ \hat B_s $ are
\begin{subequations}\label{dmkp-B1-B2}
\begin{align}
&\hat{B}_2=x\hat{A}_2+(x+n)\hat{A}_1,\label{dmkp-B2}\\
&\hat{B}_3=x\hat{A}_3+(x+n)\hat{A}_2
+v(\Delta^{-1}v_0)\Delta-v(\Delta^{-1}v)\Delta
+n v\Delta,
\end{align}
\end{subequations}
where $\hat A_j$ are given in \eqref{dmkp-Aj}.
Then, \eqref{dmkp-non-zcc} defines the  nonisopectral D$\Delta$mKP  hierarchy
\begin{equation}\label{dmkp-non-sigma-m}
v_{t_m}=\hat{\sigma}_m(v)= (\hat{B}_{m,x}+[\hat{B}_m,\hat{A}_1])\Delta^{-1},
~~  m=2,3,\cdots,
\end{equation}
where the first nonisospectral flow reads
\begin{equation}\label{dmkp-sigma-2}
\hat{\sigma}_2 =
x \hat{K}_2
+(x+n)\hat{K}_1
+v(1+2\Delta^{-1})(\ln v)_{x}+v-v^2,
\end{equation}

Ref.\cite{LZZ-SAMP-2024} shows that the isospectral and ninisospectral flows
$\{\hat{K}_m,\hat{\sigma}_s\}$ compose a Lie algebra in terms of Lie product
$\llbracket \cdot , \cdot  \rrbracket$,
where the G$\hat{\mathrm{a}}$teaux derivatives are taken with respect to $\hat v$:
\begin{equation}\label{v-hat}
\hat v= v-1.
\end{equation}
The algebraic structure is presented in Theorem 3 in \cite{LZZ-SAMP-2024}.
However, noticing that \eqref{v-hat} indicates the algebraic structure keeps invariant
if the G$\hat{\mathrm{a}}$teaux derivatives are taken with respect to $v$,
therefore we can reach to the following theorem immediately.

\begin{theorem}\label{Th-6-1}
The D$\Delta$mKP flows $\{\hat{K}_m,\hat{\sigma}_s\}_{m \ge 1, s \ge 2}$
compose a Lie algebra with the following structure
\begin{subequations}\label{K-S-alg}
\begin{align}
&\llbracket \hat{K}_m,\hat{K}_s \rrbracket=0,\\
&\llbracket \hat{K}_m,\hat{\sigma}_s \rrbracket=m ( \hat{K}_{m+s-1}+\hat{K}_{m+s-2}), \label{dmkp-Km-sigma-r}\\
&\llbracket \hat{\sigma}_m,\hat{\sigma}_s \rrbracket=(m-s) (\hat{\sigma}_{m+s-1}+\hat{\sigma}_{m+s-2}).
\end{align}
\end{subequations}
Here the G$\hat{\mathrm{a}}$teaux derivatives are taken with respect to $v$,
e.g. $\hat{K}_m^{\,\prime}[\hat{\sigma}_s]=\hat{K}_m^{\,\prime}(v)[\hat{\sigma}_s]$.
The master symmetry of the D$\Delta$mKP hierarchy is $\hat{\sigma}_2$, given in \eqref{dmkp-sigma-2},
which provides the recursion relation for the isospectral flows $\{\hat{K}_s\}$:
\begin{equation}
\hat{K}_{s+1} = \frac{1}{s} \llbracket \hat{K}_s, \hat{\sigma}_2\rrbracket - \hat{K}_s.\label{dmkp-master-sym}
\end{equation}
\end{theorem}

In addition, one can prove the following.
\begin{proposition}\label{prop-6-1}
The operators $\{\hat A_m\}$ and $\{\hat B_s\}$ satisfy the following relation
\begin{subequations}
\begin{align}
&\hat A_{m+1} = \frac{1}{m}
\left(
\hat A_m^{\,\prime}(v)[\hat \sigma_2] - \hat B_2^{\,\prime}(v)[\hat K_m] + [\hat A_m, \hat B_2]
\right)
-\hat A_m,~~m\geq 1,\\
&\hat B_{s+1} = \frac{1}{s-2}
\big(
\hat B_s^{\,\prime}(v)[\hat \sigma_2] - \hat B_2^{\,\prime}(v)[\hat \sigma_s] + [\hat B_s, \hat B_2]
\big) -\hat B_s, ~~~ s \geq 3,
\end{align}
\end{subequations}
where $\hat A_1$ is defined as \eqref{dmkp-A1}, $\hat B_2, \hat B_3$ are given in \eqref{dmkp-B1-B2}.
This also indicates that $\hat B_2$ acts as a ``master operator'' in the above recursive structures.
\end{proposition}

\subsection{Constraint of D$\Delta$mKP system to sdBurgers hierarchy}\label{sec-6-2}

To get the sdBurgers hierarchy from the D$\Delta$mKP system,
we take a linear eigenfunction constraint
\begin{equation}\label{v=psi_n}
v = \psi
\end{equation}
where we note that
\begin{equation}\label{psi-1}
\psi \to 1,~ |n|\to \infty.
\end{equation}
Under this constraint, it follows from \eqref{dmkp-iso-A1} that
\begin{equation}\label{6.14}
v_{x} = v  \Delta v,
\end{equation}
which is the sdBurgers equation \eqref{bur-k1} with $x=t_1$.
Furhter, both $\psi_{t_1}=\hat{A}_1 \psi$ and the first isospectral equation \eqref{6.9a}
are converted to
\begin{equation}
v_{t_1} = v  \Delta v,
\end{equation}
which is exactly the sdBurgers equation \eqref{bur-k1}.
Using \eqref{6.14}  one can replace all the derivatives $\partial_x^{j}v^{(n+i)}$
in the D$\Delta$mKP flows $\{\hat{K}_m\}$ and operators $\{\hat{A}_m\}$,
and then one can express them in terms of $v$ and its shifts.

Let us denote
\begin{equation}
F(v) |_{R_2}= F(v) |_{\eqref{v=psi_n},\eqref{6.14}}.
\end{equation}
Again, one should notice the asymptotic condition \eqref{psi-1}.
For example, under the constraints \eqref{v=psi_n} and \eqref{6.14},
from  \eqref{dmkp-v0-v1} and \eqref{6.3} we have
\begin{equation}
v_0 (v)|_{R_2}
= \Delta^{-1} (\ln \psi)_x |_{R_2}
= \Delta^{-1} \frac{\psi_x}{\psi} \big|_{R_2}
= \Delta^{-1} \Delta \psi
= \Delta^{-1} \Delta (\psi - 1)
= \psi - 1;
\end{equation}
and for the term $\Delta^{-1} (\ln v)_{xx}$ in \eqref{6.9b}, since $\psi_x \to 0$, we have
\begin{equation}
\Delta^{-1} (\ln v)_{xx}|_{R_2}
= \Delta^{-1} \Delta \psi_x |_{R_2}
= \psi_x |_{R_2}
= \psi \Delta \psi.
\end{equation}

Then we present the following constraint results of the D$\Delta$mKP system.

\begin{theorem}\label{Th-6-2}
Under the constraints \eqref{v=psi_n} and \eqref{6.14},
both the isospectral D$\Delta$mKP hierarchy \eqref{dmkp-iso-km}
and the time evolution \eqref{dmkp-flow-lax} for $\psi$
are converted to the combined sdBurgers hierarchy \eqref{sdBurgers-Km} with $z=v$.
\end{theorem}

This theorem can be proved in a way similar to the case for the D$\Delta$KP system in Sec.\ref{sec-3-2}.
Here we list the necessary lemmas for the proof and skip presenting the detailed proofs.

\begin{lemma}\label{lem-dmkp-sdburgers}
For the isospectral D$\Delta$mKP flows $\{\hat{K}_m(v)\}$ and master symmetry $\hat{\sigma}_2(v)$
given in \eqref{dmkp-iso-km} and \eqref{dmkp-sigma-2},
and for  the isospectral sdBurgers  flows $\{\ti{K}_m\}$ and master symmetry $\ti{\sigma}_1$ defined in \eqref{sdBurgers-Km} and \eqref{sdBurgers-sigma_1},
under the constraints \eqref{v=psi_n} and \eqref{6.14},
we have
\begin{subequations}
\begin{align}
& \hat{K}_m(v) |_{R_2}
= \ti{K}_m(\psi), \\
& \hat{\sigma}_2(v) |_{R_2}
= \hat{\ti{\sigma}}_1(\psi),
\end{align}
\end{subequations}
where
\begin{equation}
\hat{\ti{\sigma}}_1(v)
= x \ti{K}_2 + x \ti{K}_1  +\ti{\sigma}_1.\label{hati-sigma-1}
\end{equation}
In addition, for the D$\Delta$mKP operators $\hat{A}_m$ and $\hat{B}_2$, we have
\begin{subequations}
\begin{align}
& \hat A_m^{\,\prime}(v) [\hat \sigma_2(v)] |_{R_2}
= (\hat A_m |_{R_2})' (\psi)[\hat{ \ti{\sigma}}_1(\psi)],\\
& \hat B_2^{\,\prime}(v)[\hat K_m(v)] |_{R_2}
= (\hat B_2 |_{R_2})' (\psi)[\ti{K}_m (\psi)].
\end{align}
\end{subequations}
\end{lemma}

\begin{lemma}\label{lem-dmkp-sdburgers-3}
Define flows $\{\tau_m(\psi)\}$ by
\begin{subequations}
\begin{align}
&\tau_1 (\psi)
=-\psi^{(n+1)} \psi^{(n)} + \psi^{(n)},\\
&\tau_{m+1} (\psi)
=\tau_m^{\,\prime}[\hat{\ti{\sigma}}_1 (\psi)]
- \hat B_2 |_{R_2} \tau_m(\psi),~~m \ge 1.
\end{align}
\end{subequations}
where $\hat B_2$ is the operator given in \eqref{dmkp-B2} and $\hat{\ti{\sigma}}_1(\psi)$ is defined in \eqref{hati-sigma-1}.
Then, for the operator $\hat A_s$ and operator $\mathcal{M}$ defined as
\begin{equation}
\mathcal{M} = \hat{A}_1 - \partial_x,
\end{equation}
we have
\begin{subequations}\label{M,An-tau}
\begin{align}
&\mathcal{M} |_{R_2} \tau_m
= 0,\\
&\hat A_s |_{R_2} \tau_m
= \tau_m^{\,\prime}(\psi)[\ti{K}_s(\psi)].
\end{align}
\end{subequations}
\end{lemma}

\begin{lemma}\label{lem-dmkp-sdburgers-4}
Under the constraint \eqref{v=psi_n}, we have isospectral sdBurgers flows
\begin{equation}
\hat A_m(v) |_{R_2} \psi
= \ti{K}_m (\psi).
\end{equation}
In addition, we have
\begin{equation}
\hat B_2(v) |_{R_2} \psi
= \hat{\ti{\sigma}}_1(\psi) + \tau_1(\psi).
\end{equation}
\end{lemma}

Let end this section with the following remark.
\begin{remark}\label{Rem-6-1}
It is well known that their is a Miura transformation between
the D$\Delta$KP and D$\Delta$mKP hierarchies, which reads (e.g.\cite{CZZ-SAMP-2021,JC-MPLB-2018,TK-CSF-2000})
\begin{equation}\label{trans-u-v}
u = \Delta^{-1} (\ln v)_x + 1- v.
\end{equation}
However, when $v$ satisfies the sdBurgers equation \eqref{6.14} with asymptotic property
$v\to 1~(|n|\to \infty)$, it turns out that $u=0$ from \eqref{trans-u-v},
which means the transformation \eqref{trans-u-v}
becomes trivial under the constraint \eqref{v=psi_n}.
\end{remark}

\section{Conclusions}\label{sec-7}

In this paper we have investigated three eigenfunction constraints of
the D$\Delta$KP  hierarchy and the D$\Delta$mKP hierarchy.
We revisited the squared eigenfunction symmetry constraint of the D$\Delta$KP  hierarchy,
which gives rise to the sdAKNS hierarchy
from the time evolutions of the eigenfunction and  the adjoint eigenfunction.
This fact has been proved before in \cite{CDZ-JNMP-2017} by using the recursion operator of the sdAKNS hierarchy,
but in this paper we proved it by using the recursive algebraic structures of the involved flows
generated by their master symmetries.
These recursive structures are useful.
In light of the linear eigenfunction constraint \eqref{u=delta phi},
both the D$\Delta$mKP hierarchy  \eqref{dkp-Km} and the time  evolution \eqref{dkp-iso-lax-t}
of the eigenfunction are converted to the combined the sdBurgers hierarchy \eqref{sdBurgers-Km}.
To prove this, in Sec.\ref{sec-4} we formulated the (combined) isospectral and nonisospectral sdBergurs hierarchies,
their algebraic structure and the related master symmetry.
Note that in the formulation one needs to take care of the asymptotic condition \eqref{z-asym}
in dealing with discrete `integration'.
The constraint results are proved by comparing their recursive algebraic structures.
In addition, we also showed that the linear eigenfunction constraint \eqref{v=psi_n}
leads to the same combined sdBurgers hierarchy \eqref{sdBurgers-Km}
from the D$\Delta$mKP hierarchy and the time  evolution \eqref{dmkp-flow-lax}
of the eigenfunction.
The proof can be implemented along the same line.
Note that the D$\Delta$mKP system allows two different squared eigenfunction constraints:
one yields  the relativistic Toda hierarchy (see \cite{CZZ-SAMP-2021})
and the other yields the semi-discrete CLL hierarchy (see \cite{LZZ-SAMP-2024}),
both of which can be reduced to the sdBurgers hierarchy.
If one revisits the reductions in \cite{CZZ-SAMP-2021} by  strictly considering the asymptotic condition
of the involved variables,
then the sdBurgers hierarchy reduced from \cite{CZZ-SAMP-2021} will be the same as
what we have obtained in this paper.
The research of this paper shows more interesting mathematical structures in differential-difference case
as well as applications of master symmetries.

\vskip 20pt
\subsection*{Data availability statement}

No new data were created or analysed in this study.


\subsection*{Conflict of interest statement }

The authors declare that they have no conflict of interest.

\vskip 20pt
\subsection*{Acknowledgments}

This project is supported by the National Natural Science Foundation of China (Grant No. 12271334).

\end{document}